\documentclass[aps,prl,reprint,superscriptaddress,showpacs]{revtex4-2}

\usepackage[utf8]{inputenc}  
\usepackage[T1]{fontenc}     
\usepackage[american]{babel}  
\usepackage[svgnames,dvipsnames]{xcolor}

\usepackage{amsmath,amssymb,amsthm,bm,mathtools,amsfonts,mathrsfs,bbm,dsfont,times,appendix,multirow,color,cases,tikz} 

\usepackage[ruled]{algorithm2e}

\usepackage{verbatim}
\usepackage{multirow}
\usepackage{braket}
\usepackage[normalem]{ulem}

\usepackage{hyperref}
\hypersetup{colorlinks=true, citecolor=FireBrick, urlcolor=blue, linkcolor=blue}

\newtheorem{observation}{Observation}
\newtheorem{corollary}{Corollary}

\newcommand{\tr}{{\mathrm{tr}}}

\newcommand{\eins}{\mathbbm{1}}

\renewcommand{\vr}{\ensuremath{\varrho}}

\begin{document}
\title{Detecting high-dimensional entanglement by randomized product projections}	
\author{Jin-Min Liang}
\affiliation{School of Mathematics and Statistics, Beijing Technology and Business University, Beijing 100048, China}
\author{Shuheng Liu}
\affiliation{State Key Laboratory for Mesoscopic Physics, School of Physics, Frontiers Science Center for Nano-optoelectronics, $\&$ Collaborative Innovation Center of Quantum Matter, Peking University, Beijing 100871, China}
\author{Shao-Ming Fei}
\affiliation{School of Mathematical Sciences, Capital Normal University, Beijing 100048, China}
\author{Qiongyi He}
\email{qiongyihe@pku.edu.cn}
\affiliation{State Key Laboratory for Mesoscopic Physics, School of Physics, Frontiers Science Center for Nano-optoelectronics, $\&$ Collaborative Innovation Center of Quantum Matter, Peking University, Beijing 100871, China}
\affiliation{Collaborative Innovation Center of Extreme Optics, Shanxi University, Taiyuan, Shanxi 030006, China}
\affiliation{Hefei National Laboratory, Hefei 230088, China}

\date{\today}
\begin{abstract}
    The characterization of high-dimensional entanglement plays a crucial role in the field of quantum information science. Conventional entanglement criteria measuring coherent superpositions of multiple basis states face experimental bottlenecks on most physical platforms due to limited multi-channel control. Here, we introduce a practically efficient detection strategy based on randomized product projections. We show that the first-order moments of such projections can be used to estimate entanglement fidelity, thereby enabling practical and efficient certification of the Schmidt number in high-dimensional bipartite systems. By constructing optimal observables, it is sufficient to merely measure a single basis state, substantially reducing experimental overhead. Moreover, we present an algorithm to obtain a lower bound of the Schmidt number with a high confidence level from a limited number of experimental data. Our results open up resource-efficient experimental avenues to detect high-dimensional entanglement and test its implementations in modern information technologies.
\end{abstract} 

\maketitle

\textit{Introduction}---Quantum entanglement is a valuable resource in various quantum applications from quantum cryptography~\cite{scarani2009security,xu2020secure,pirandola2020advances}, quantum metrology~\cite{giovannetti2011advances,degen2017quantum,toth2014quantum}, and quantum computation~\cite{amico2008entanglement,zhao2025entanglement}. The preparation, manipulation, and characterization of quantum entanglement are important in quantum information science. Beyond qubit systems, high-dimensional entanglement has attracted tremendous interest from bipartite~\cite{erhard2020advances,dada2011experimental,nape2021measuring,tabia2022bell,alessandro2025semidefinite,tavakoli2024enhanced} to multipartite systems~\cite{cobucci2024detecting,liu2024nonlinear}. The preparation of high-dimensional entangled states is feasible in different experimental platforms such as polar molecules~\cite{yan2013observation}, cold atomic ensembles~\cite{parigi2015storage}, trapped ions~\cite{senko2015realization}, and photonic systems~\cite{mair2001entanglement,wang2012terabit,krenn2016automated,valencia2020high}. High-dimensional entanglement in quantum communication brings several benefits by increasing information capacity~\cite{vertesi2010closing} and improving noise resistance~\cite{cerf2002security,zhu2021high}. These advantages have sparked significant interest in characterizing high-dimensional entanglement.

Entanglement witness provides an experimentally feasible way to detect entanglement of high-dimensional states~\cite{horodecki1996separability,guhne2009entanglement,bavaresco2018measurements,zhang2024analyzing}. The typical Schmidt number ($\mathrm{SN}$) witness is defined by the fidelity $F(\vr)=\tr(\vr|\phi_d^{+}\rangle\langle\phi_d^{+}|)$ between the target state $\vr$ and a maximally entangled state $|\phi_d^{+}\rangle=\frac{1}{\sqrt{d}}\sum_{j=0}^{d-1}|jj\rangle$. Nevertheless, the direct measurement of the fidelity-based SN witness has become increasingly challenging as the dimension grows, mainly due to the requirement of $d(d+1)$ projections~\cite{morelli2023resource}. To reduce the experimental cost, a sequence of $\mathrm{SN}$ criteria has been introduced~\cite{morelli2023resource} by estimating a lower bound of the fidelity $F(\vr)$ from mutually unbiased bases (MUBs)~\cite{wootters1989optimal,ivonovic1981geometrical}. This strategy is quite powerful in verifying high $\mathrm{SN}$, which requires only a small number of MUBs.

The MUBs method exhibits inherent trade-offs between estimation accuracy of fidelity $F(\vr)$ and the required number of implemented MUBs~\cite{morelli2023resource}. Fewer projections ease experiments but weaken fidelity's lower bound, while increased accuracy demands more MUB measurements. Each MUB measurement involves projections onto $d$ basis states, and hence the experimenter needs to simultaneously and precisely control multiple channels. However, on many platforms such as photonic time-frequency domain platforms~\cite{martin2017quantifying,friis2019entanglement}, available measurements are limited to measuring coherent superpositions of two basis states. On the other hand, performing a large number of different measurements is infeasible, and thus experimental data obtained from the lab is limited. These challenges motivate us to develop entanglement certification protocols with limited control and limited experimental data.

\begin{figure}[ht]
    \centering
    \includegraphics[scale=0.3]{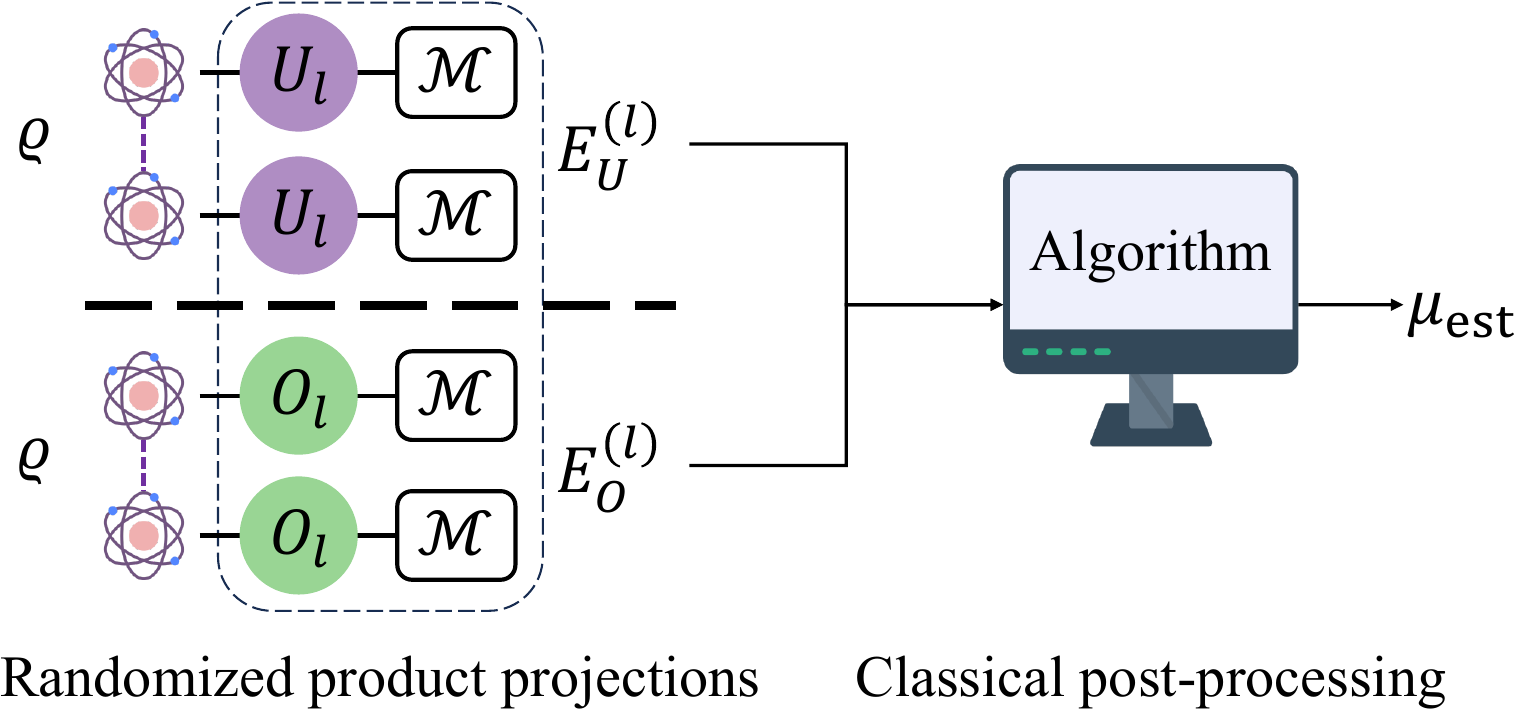}
    \caption{Illustration of randomized product projections for estimating the Schmidt number of a bipartite state $\vr$. The first stage implements local operations $U_l^{\otimes2}$ and $O_l^{\otimes2}$ and then measures on the well-constructed observable $\mathcal{M}^{\otimes2}$. The random unitaries $U_l$ and random orthogonal matrices $O_l$ are sampled independently from the Haar measure on unitary and orthogonal groups~\cite{goodman2000representations,collins2006integration,gross2007evenly,dankert2009exact}. From the obtained expectation values $E_{U}^{l}$ and $E_{O}^{(l)}$, the algorithm outputs the Schmidt number estimation $\mu_{\mathrm{est}}$.}
    \label{Fig1}
\end{figure}

In this work, we address the above issues by introducing randomized product projections, as illustrated in Fig.~\ref{Fig1}. In particular, Alice and Bob rotate the state under the same unitary or orthogonal matrix, followed by measurements on product observables. With the obtained first-order moments, we propose a Schmidt number criterion of a bipartite high-dimensional state by directly estimating the fidelity $F(\vr)$ in a randomized way. Crucially, the proposed criterion is endowed with robustness against orthogonal noise since the orthogonal invariance of the first-order moments. We construct an optimal observable $\mathcal{M}$ as a single projection and constrain the required measurements to be a single basis state projection, which is particularly efficient for experimental platforms such as integrated-optics platforms~\cite{schaeff2015experimental} and practical for time-frequency-domain platforms~\cite{martin2017quantifying,friis2019entanglement}. Considering the limited experimental data, we develop an algorithm to obtain a lower bound of the Schmidt number with a high confidence level from a finite number of projections. We theoretically establish that the total number of required random operations scales as $\mathcal{O}(d^2)$ for general quantum states, which is consistent with the direct fidelity estimation. However, when the target state is a maximally entangled state subject to typical preparation noise, the number of random operations reduces to $\mathcal{O}(1)$, independent of local dimension $d$. In this case, our protocol outperforms both direct fidelity estimation and the MUB approach, which needs $\mathcal{O}(md)$ settings, where $m$ is the number of MUBs used. Numerical results show that $60$ random operations suffice for mixed states of local dimensions $ d=20, 30, 40, 50$ to yield a more accurate Schmidt number estimation than previous approaches~\cite{morelli2023resource,liu2023characterizing,wyderka2023probing}. Consequently, our scheme significantly reduces the resource overhead for entanglement analysis of high-dimensional bipartite quantum states.

\textit{Schmidt number and fidelity-based witness}---The Schmidt number ($\mathrm{SN}$) of a pure quantum state $|\psi\rangle$ is defined by Schmidt rank ($\mathrm{SR}$), i.e., the number of nonzero Schmidt coefficients in their Schmidt decomposition~\cite{terhal2000schmidt}, $|\psi\rangle=\sum_{l=1}^{\mu}\sqrt{s_l}|a_lb_l\rangle$, where $s_l>0$, $\sum_{l}s_l=1$, $\mu$ denotes the $\mathrm{SR}$, and $\{a_l\}$, $\{b_l\}$ are respectively orthonormal states. For mixed state $\vr$, its $\mathrm{SN}$ is defined by~\cite{terhal2000schmidt}, $\mathrm{SN}(\vr)=\min_{\mathcal{D}(\vr)}\{\mu_{\max}:\vr=\sum_{l}p_l|\psi_l\rangle\langle\psi_l|,\mu_{\max}=\max_{l}\mu_l\}$, where $\mu_l$ is the $\mathrm{SR}$ of pure state $|\psi_l\rangle$ and $\mathcal{D}(\vr)$ denotes the set of all pure-state decomposition of the state $\vr$. A larger $\mathrm{SN}$ in a bipartite state implies a higher degree of entanglement.

A typical way to determine the $\mathrm{SN}$ of a state $\vr$ is to estimate the fidelity $F(\vr)=\tr(\vr|\phi_d^{+}\rangle\langle\phi_d^{+}|)$ between $\vr$ and a maximally entangled state $|\phi_d^{+}\rangle=\sum_{j=0}^{d-1}|jj\rangle/\sqrt{d}$. For any $\vr$ with $\mathrm{SN}(\vr)=\mu$, $F(\vr)\leq\mu/d$~\cite{terhal2000schmidt}. This inequality can be transformed into an $\mathrm{SN}=\mu+1$ witness, $\mathcal{W}_\mu=\mu\eins_d^{\otimes2}/d-|\phi_d^{+}\rangle\langle\phi_d^{+}|$~\cite{terhal2000schmidt,sanpera2001schmidt,guhne2021geometry,wyderka2023construction} such that $\tr(\vr\mathcal{W}_\mu)\geq0$ for all $\vr$ with $\mathrm{SN}(\vr)\leq\mu$, where $\eins_d$ denotes the identity. Thus, if $\tr(\vr\mathcal{W}_\mu)<0$, $\vr$ has at least $\mathrm{SN}(\vr)=\mu+1$.
	
\textit{Schmidt number criterion via randomized product projections}---Consider a two-qudit state $\vr$ shared by Alice and Bob. To detect the $\mathrm{SN}$ of $\vr$, we measure $\vr$ with randomized product projections~\cite{van2022hardware,imai2024collective}. The measurement process includes two stages. The first stage allows Alice and Bob to perform measurements with any product observables $\mathcal{M}^{\otimes2}$ after applying the same random unitary $U^{\otimes2}$ on the state $\vr$. The expectation values for a fixed unitary $U$ is, $E_U(\vr)=\tr[U^{\otimes2}\vr U^{\dag\otimes2}\mathcal{M}^{\otimes2}]$. This randomized measurement strategy has been used to classify topological order~\cite{van2022hardware} and analyze multipartite entanglement~\cite{imai2024collective}. When Alice and Bob are on the same platform, the local operation $U^{\otimes2}$ can be implemented by passing both parties through a single $U$ setup at different times, thereby reducing the operational costs while simultaneously mitigating imperfections of the unitary.
	
Then, we use a random orthogonal matrix $O$ satisfying $O^{\top}O=OO^{\top}=\eins_d$, where $\top$ denotes the transpose operation. In this scenario, Alice and Bob implement the same orthogonal matrix $O^{\otimes2}$ on the state $\vr$ and perform measurements with product observables $\mathcal{M}^{\otimes2}$, similar to the unitary process. The expectation value for a fixed orthogonal matrix $O$ is, $E_O(\vr)=\tr[O^{\otimes2}\vr O^{\top\otimes2}\mathcal{M}^{\otimes2}]$. The implementation of $O^{\otimes2}$ can also be simplified using a strategy similar to that of $U^{\otimes2}$.
	
After sampling random unitary $U$ and orthogonal matrix $O$ respectively according to the Haar measure on unitary and orthogonal groups~\cite{goodman2000representations,collins2006integration,gross2007evenly,dankert2009exact}, we consider two first-order moments,
    \begin{align}\label{II:eq1}
        \mathcal{R}(\vr)=\int dUE_U(\vr),~
        \mathcal{Q}(\vr)=\int dOE_O(\vr).
    \end{align}
The experimental feasibility of our protocol critically depends on the rank $r$ (i.e. the number of nonzero eigenvalues) of the observable $\mathcal{M}$. When the observable has rank $r$, estimating the expectation of $E_U(\vr)$ and $E_O(\vr)$ requires measurements across $r$ basis states, simultaneously controlling $r$ experimental channels in physical platforms. To minimize this technical complexity, we construct an optimal observable with $r=1$, $\mathcal{M}=|j\rangle\langle j|$, where $\{|j\rangle\}_{j=0}^{d-1}$ denotes the computational basis in a $d$-dimension Hilbert space. This construction reduces the number of channel controls to one, independent of local dimensions $d$. Such low-rank implementations circumvent multi-channel challenges, making our protocol more easily implementable in practical physical platforms.

Using the optimal observable $\mathcal{M}=|j\rangle\langle j|$, we then have
\begin{align}
    \mathcal{R}(\vr)&=\frac{1+\tr(\vr\mathbb{S})}{d(d+1)},
    \mathcal{Q}(\vr)=\frac{1+\tr\left(\vr\mathbb{S}\right)+\tr\left(\vr\mathbb{S}^{\top_B}\right)}{d(d+2)},\label{II:eq6}
\end{align}
by using the Schur–Weyl duality in the unitary and orthogonal group~\cite{goodman2000representations,collins2006integration,gross2007evenly,dankert2009exact}. Here, $\mathbb{S}=\sum_{j,k=0}^{d-1}|jk\rangle\langle kj|$ is the $\mathrm{SWAP}$ operator and $\top_B$ denotes the partial transpose to the subsystem B. See detailed proof in Appendix~\ref{AppendixA}.
 
Now, we present our result showing that the moments $\mathcal{R}(\vr)$ and $\mathcal{Q}(\vr)$ in Eq.~(\ref{II:eq6}) can be used to calculate the fidelity $F(\vr)=\tr(\vr|\phi_d^{+}\rangle\langle\phi_d^{+})$ and characterize high-dimensional entanglement.
\begin{observation}\label{ob1}
    For any bipartite state $\vr$ of equal dimension $d$, the fidelity $F(\vr)$ is given by
    \begin{align}\label{II:eq2}
        F(\vr)=(d+2)\mathcal{Q}(\vr)-(d+1)\mathcal{R}(\vr).
    \end{align}
    The Schmidt number of $\vr$ is lower bounded by
    \begin{align}\label{II:eq3}
        \mu=
        \begin{cases}
            1, & F(\vr)=0,\\
            \lceil dF(\vr)\rceil, & F(\vr)\in(0,1],\\
        \end{cases}
    \end{align}
    where the $\lceil\cdot\rceil$ is ceil function.
\end{observation}
Equation~(\ref{II:eq2}) provides a formula to estimate the fidelity from the first-order moments $\mathcal{R}(\vr)$ and $\mathcal{Q}(\vr)$ obtained by randomized product projection. Then, using obtained fidelity we obtain a $\mathrm{SN}$ estimation of $\vr$ based on Eq.~(\ref{II:eq3}). Detailed proof is presented in Appendix~\ref{AppendixA}, which utilizes the Eq.~(\ref{II:eq6}), the relation $d|\phi_d^{+}\rangle\langle\phi_d^{+}|=\mathbb{S}^{\top_B}$, as well as the property of the fidelity-based witness $\mathcal{W}_\mu$. We remark that Eq.~(\ref{II:eq3}) is equivalent to the fidelity-based witness $\mathcal{W}_\mu$. This equivalence implies that our criterion and $\mathcal{W}_\mu$ possess the same capability to detect $\mathrm{SN}$. Moreover, Eq.~(\ref{II:eq2}) is invariant under the transformation $O\otimes O$. This indicates that our approach is robust against orthogonal $O\otimes O$ noise in the measurement stage. As a comparison, the previous methods, such as MUBs and direct fidelity estimation, require well-controlled measurements and thus are fragile to the measurement noise model.

Note that randomized measurements have been extensively investigated for entanglement characterization~\cite{imai2021bound,wyderka2023probing,liu2023characterizing,yi2025certifying,giovanni2025entanglement} in a reference-frame-independent manner. However, the randomized product projection in observation~\ref{ob1} requires both parties to synchronise their measurement bases. Consequently, our protocol is less suited for scenarios involving distant laboratories.

Observation~\ref{ob1} provides a theoretical criterion and naturally requires us to estimate the moments $\mathcal{R}(\vr)$ and $\mathcal{Q}(\vr)$ exactly with infinite samples of random operations, which is impossible in a practical experimental scenario. Therefore, we next introduce an algorithm to obtain an $\mathrm{SN}$ estimation from the measured experimental data with limited projections.

\textit{Experimental data acquisition and the number of randomized product projections}---For the estimation of the integral over unitary and orthogonal groups, we independently sample $N$ unitaries $\{U_l\}_{l=1}^{N}$ and $N$ orthogonal matrices $\{O_l\}_{l=1}^{N}$ from a unitary $2$-design and orthogonal $2$-design~\cite{goodman2000representations,collins2006integration,gross2007evenly,dankert2009exact}. After performing measurements on a rank-one observable $\mathcal{M}^{\otimes2}=|j\rangle\langle j|^{\otimes2}$, we obtain a set of the expectation values $\mathcal{S}_U^{(N)}=\{E_U^{(l)}(\vr)\}_{l=1}^{N}$ and $\mathcal{S}_O^{(N)}=\{E_O^{(l)}(\vr)\}_{l=1}^{N}$, where the expectations values are
\begin{align}
    E_U^{(l)}(\vr)&=\tr[U_l^{\otimes2}\vr U_l^{\dag\otimes2}\mathcal{M}^{\otimes2}],\label{III:eq1}\\
    E_O^{(l)}(\vr)&=\tr[O_l^{\otimes2}\vr O_l^{\top\otimes2}\mathcal{M}^{\otimes2}].\label{III:eq2}
\end{align}
The data set $\mathcal{S}_U^{(N)}$ and $\mathcal{S}_O^{(N)}$ can then be used to obtain an estimator of the fidelity $F(\vr)$.

The number of randomized product projections depends on the number of random operations, $N$. Note that measuring the observable $\mathcal{M}^{\otimes2}$ on states $U_l^{\otimes2}\vr U_l^{\dag\otimes2}$ and $O_l^{\otimes2}\vr O_l^{\top\otimes2}$ equivalent to perform the global product projection $|e^{l}\rangle\langle e^{l}|^{\otimes2}$ and $|u^{l}\rangle\langle u^{l}|^{\otimes2}$ on $\vr$, respectively, where $\{|e^{l}\rangle=U_l^{\dag}|j\rangle\}_{l=1}^{N}$ and $\{|u^{l}\rangle=O_l^{\top}|j\rangle\}_{l=1}^{N}$. As a result, the number of projections for estimating the fidelity $F(\vr)$ is $2N$. 

In Appendix~\ref{AppendixB}, we show that in the worst case the number of random operations scales as $N=\mathcal{O}(d^2)$. For maximally entangled states under a depolarizing, dephasing, or random noise channel, $N=\mathcal{O}(1)$ is a constant independent of local dimension $d$. This scaling matches the numerical results in the next section. As a comparison, the scaling of the MUBs method~\cite{morelli2023resource} is $\mathcal{O}(md)$, where $m\in[2,d+1]$ is the number of MUBs. Thus, our randomized projection is significantly efficient for large local dimensions.

\textit{Classical post-processing algorithm for the Schmidt number estimation}---With the data set $\mathcal{S}_U^{(N)}$ and $\mathcal{S}_O^{(N)}$ in hand, the natural estimator of the fidelity is inferred via Eq.~(\ref{II:eq2}),
\begin{align}\label{IV:eq1}
    F_{e}(\vr)=(d+2)\mathcal{Q}_e(\vr)-(d+1)\mathcal{R}_e(\vr),
\end{align}
where the estimated moments $\mathcal{R}_e(\vr)=\mathbb{E}[\mathcal{S}_U]$ and $\mathcal{Q}_e(\vr)=\mathbb{E}[\mathcal{S}_O]$ are defined as the mean values of $\mathcal{S}_U$ and $\mathcal{S}_O$, respectively.
 
To estimate the statistical distribution of $F_{e}(\vr)$, a direct approach is to calculate many $F_{e}(\vr)$ from several data sets. However, obtaining more data sets is extremely challenging since additional experimental runs are required to estimate the expectation values of observables. In the following, we present an algorithm to estimate an $\mathrm{SN}$ lower bound by only using a single data set $\mathcal{S}_U$ and $\mathcal{S}_O$.

For each expectation values $E_U^{(l)}(\vr)$ and $E_O^{(l)}(\vr)$ defined in Eqs.~(\ref{III:eq1},\ref{III:eq2}), we begin by estimating the quantity $F_{e}^{(l)}(\vr)=(d+2)E_O^{(l)}(\vr)-(d+1)E_U^{(l)}(\vr)$ via Eq.~(\ref{II:eq2}). Then, we collect the data set $\tilde{F}_{e}(\vr)=\{F_{e}^{(l)}(\vr)\}_{l=1}^{N}$ and construct the confidence interval of $F(\vr)$, $[F_{\mathrm{lb}},F_{\mathrm{ub}}]$, where the lower bound and upper bound are
\begin{align}\label{IV:eq2}
    F_{\mathrm{lb}}&=\mathbb{E}\left[\tilde{F}_{e}(\vr)\right]-\frac{st_{\alpha,N-1}}{\sqrt{N}},\\
    F_{\mathrm{ub}}&=\mathbb{E}\left[\tilde{F}_{e}(\vr)\right]+\frac{st_{\alpha,N-1}}{\sqrt{N}}.
\end{align}
Here $\mathbb{E}[\tilde{F}_{e}(\vr)]$ and $s$ are the sample mean and sample standard deviation of the data set $\tilde{F}_{e}(\vr)$, respectively. The positive factor $t_{\alpha,N-1}$ is the upper percentage point of the $t$-distribution with $N-1$ degrees of freedom~\cite{student1908probable} and the confidence level (CL) $\alpha\in(0,1)$. See details in Appendix~\ref{AppendixC}. Note that $\mathbb{E}[\tilde{F}_{e}(\vr)]$ is an unbiased estimator of $F(\vr)$ shown in Appendix~\ref{AppendixB}. Finally, we use the lower bound $F_{\mathrm{lb}}$ as an estimator of $F(\vr)$ and thus obtain a $\mathrm{SN}$ estimator $\mu_{\mathrm{est}}$ based on Eq.~(\ref{II:eq3}) by replacing $F(\vr)$ with $F_{\mathrm{lb}}$.

The basic idea of our algorithm is that we treat the fidelity $F(\vr)$ as an unknown mean of a population. After sampling random unitaries and orthogonal matrices, we estimate the confidence interval of the unknown mean by using statistical methods~\cite{montgomery2010applied}. We have several remarks on the algorithm. First, confidence intervals are random quantities varying from sample to sample. Whether the constructed confidence interval covers the true fidelity $F(\vr)$ or not depends on the CL. A CL $\alpha$ implies that $\alpha\times100\%$ confidence intervals contain $F(\vr)$ by repeated sampling. On the other hand, a trade-off exists between the CL and the estimated precision of the $\mathrm{SN}$. A higher CL induces a wider confidence interval of the fidelity $F(\vr)$, thus corresponding to a smaller lower bound of the $\mathrm{SN}$ for a fixed $N$. For estimating the $\mathrm{SN}$, a higher CL is desired since we need to ensure that the obtained confidence interval covers the fidelity $F_{\mathrm{lb}}$ with greater certainty.

Moreover, when the number of samples is sufficiently large (often empirically taken as $N\geq30$), the algorithm is valid based on the central limit theorem~\cite{montgomery2010applied}. For the small sample size, using the $t$-distribution to construct a confidence interval assumes that the unknown population is normal. However, the $t$ distribution-based confidence interval is relatively robust to this assumption~\cite{montgomery2010applied}. See examples in Appendix~\ref{AppendixD}.

\textit{Entanglement with depolarizing noise}---To test the $\mathrm{SN}$ estimation algorithm, we suppose that the device produces the maximally entangled state $|\phi_d^{+}\rangle$ but under the depolarizing noise~\cite{nielsen2010quantum},
\begin{align}
    \vr_{\mathrm{iso}}^{v}&=v|\phi_d^+\rangle\langle\phi_d^+|+[(1-v)/d^2]\eins_d^{\otimes2},
\end{align}
where the parameter $v\in[0,1]$. The fidelities of the state $\vr_{\mathrm{iso}}^{v}$ is $F(\vr_{\mathrm{iso}}^{v})=v+(1-v)/d^2$ and its Schmidt number is at least $\mu+1$ if and only if $v>(\mu d-1)/(d^2-1)$~\cite{terhal2000schmidt}.

\begin{figure}[]
    \centering
    \includegraphics[scale=0.47]{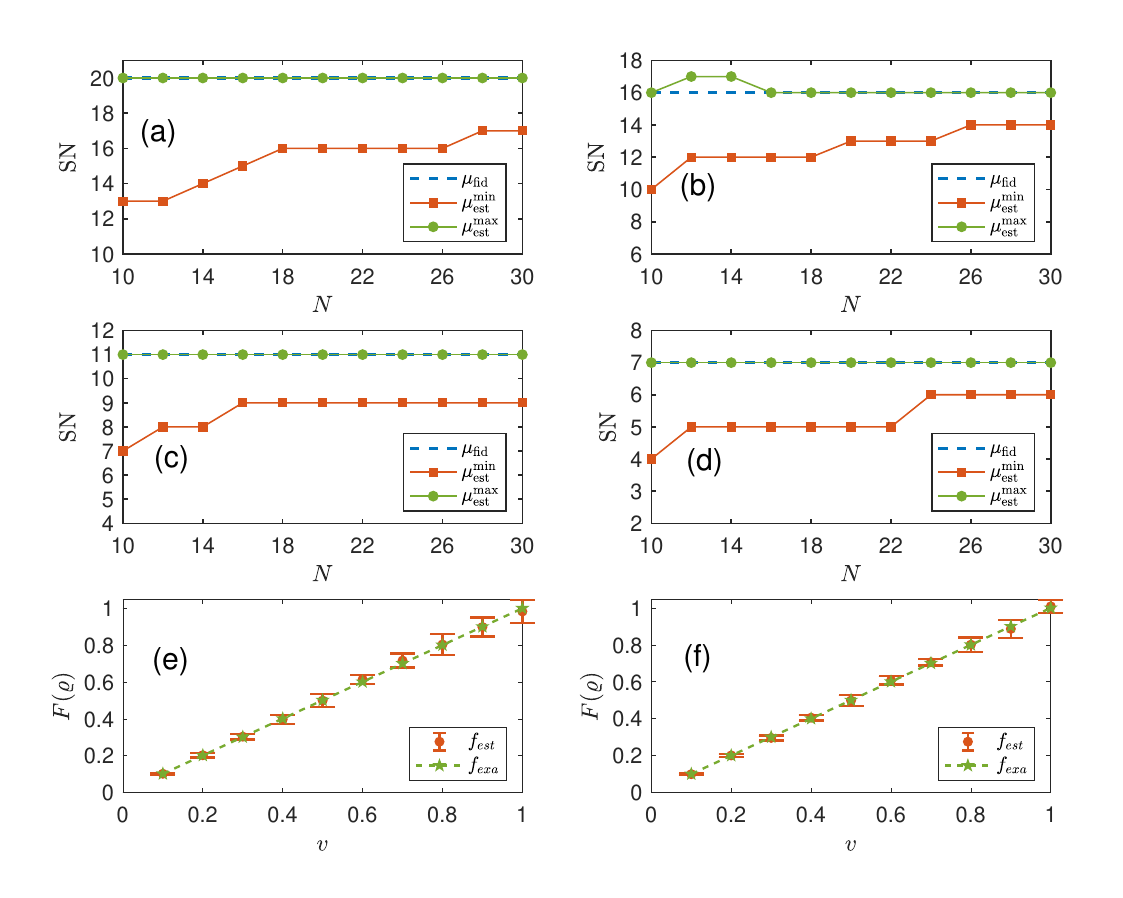}
    \caption{(a-d) shows the distribution of estimated $\mathrm{SN}$ lower bound $\mu_{\mathrm{est}}$ of the state $\vr_{\mathrm{iso}}^v$ with $d=20$ as a function of the number of $N$, where (a) $v=0.95$, (b) $v=0.77$, (c) $v=0.52$, and (d) $v=0.3$. For each $N$, the expectation values of product observables are exactly calculated, and the CL is $99.9\%$. Running the $\mathrm{SN}$ estimation algorithm $N_{\mathrm{iter}}=500$ times, we collect $N_{\mathrm{iter}}$ estimators, $\{\mu_{\mathrm{est}}^{(l)}\}_{l=1}^{N_{\mathrm{iter}}}$, in which the minimal and maximal values are denoted as $\mu_{\mathrm{est}}^{\min}$ and $\mu_{\mathrm{est}}^{\max}$. $\mu_{\mathrm{fid}}$ denotes the $\mathrm{SN}$ estimated by the fidelity-based witness. (e,f) shows the error bar of the fidelity estimation for the state $\vr_{\mathrm{iso}}^{v}$ with CL $99.9\%$ and $N=30$ random operations. (e) $d=20$ and (f) $d=30$. $f_{\mathrm{est}}$ and $f_{\mathrm{exa}}$ denote the estimated and exact fidelities. The expectation values are calculated exactly.}
    \label{Fig2}
\end{figure}

In Fig.~\ref{Fig2}(a-d), we present numerical results and infer the relation between the estimated $\mathrm{SN}$ $\mu_{\mathrm{est}}$ and the number of random operations $N$ with CL $99.9\%$. Iterative running the $\mathrm{SN}$ estimation algorithm $N_{\mathrm{iter}}=500$ times for a fixed $N$, we obtain $N_{\mathrm{iter}}$ estimators, $\{\mu_{\mathrm{est}}^{(l)}\}_{l=1}^{N_{\mathrm{iter}}}$. Define the maximal and minimal estimation error for a fixed $N$ as
\begin{align}
    \mathcal{E}_{\max}(N)=\mu_{\mathrm{fid}}-\mu_{\mathrm{est}}^{\min},
    \mathcal{E}_{\min}(N)=\mu_{\mathrm{fid}}-\mu_{\mathrm{est}}^{\max},
\end{align}
where $\mu_{\mathrm{est}}^{\min}$ and $\mu_{\mathrm{est}}^{\max}$ are the minimal and maximal value of the data set $\{\mu_{\mathrm{est}}^{(l)}\}_{l=1}^{N_{\mathrm{iter}}}$, and $\mu_{\mathrm{fid}}$ is the $\mathrm{SN}$ estimated by the fidelity-based witness. Numerical analysis shows that $\mathcal{E}_{\max}(N)$ is nearly monotone as the increase of $N$, such as $\mathcal{E}_{\max}(10)=7$ and $\mathcal{E}_{\max}(30)=3$ from Fig.~\ref{Fig2}(a). It can be found that the best estimation error $\mathcal{E}_{\min}(N)=0$ for $N\in[16,30]$. Moreover, $\mathcal{E}_{\max}(30)=3,2,2,1$ for $v=0.95,0.77,0.52,0.3$, which indicates the higher reliability of the estimation protocol under high-noise conditions. Fig.~\ref{Fig2}(e,f) further shows this point by observing that the error bars of the fidelity estimate decrease with increasing noise parameter $v$.

\begin{figure}[]
    \centering
    \includegraphics[scale=0.55]{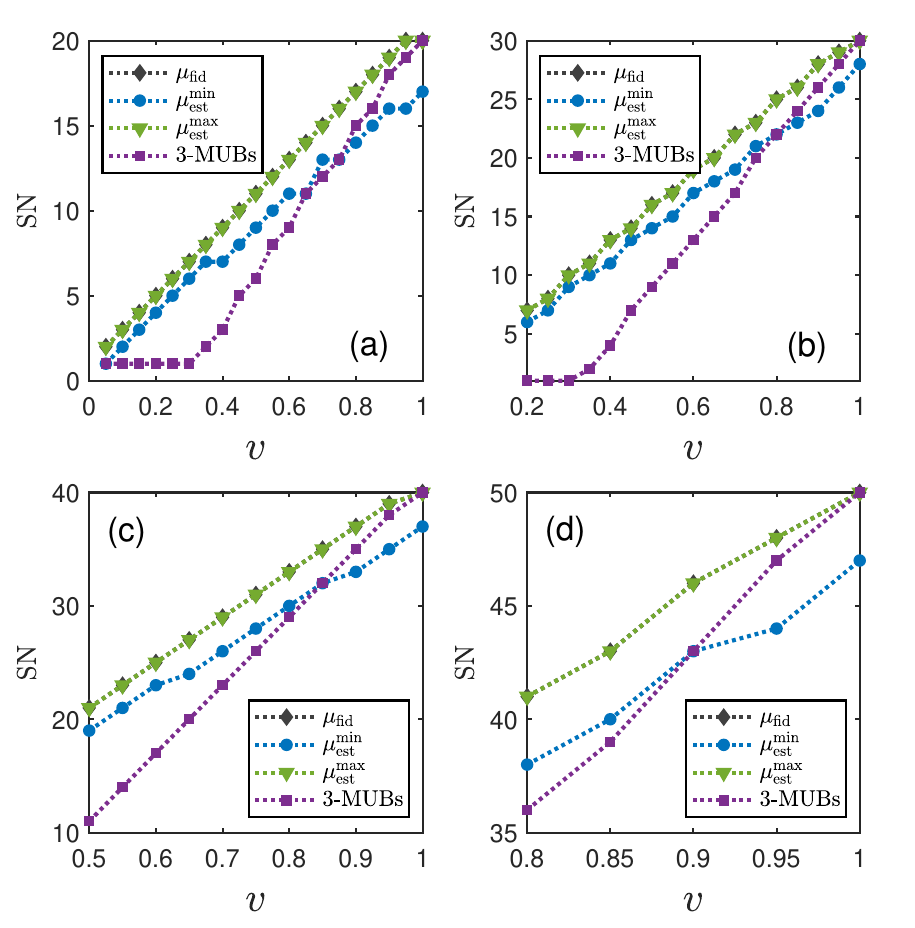}
    \caption{Results for the states $\vr_{\mathrm{iso}}^v$ with (a) $d=20$, (b) $d=30$, (c) $d=40$, and $d=50$ by fixing $N=30$ and the CL $99.9\%$. After performing the $\mathrm{SN}$ estimation algorithm $N_{\mathrm{iter}}=500$ times, we collect $N_{\mathrm{iter}}$ estimators, $\{\mu_{\mathrm{est}}^{(l)}\}_{l=1}^{N_{\mathrm{iter}}}$, in which the minimal and maximal values are denoted as $\mu_{\mathrm{est}}^{\min}$ and $\mu_{\mathrm{est}}^{\max}$. $\mu_{\mathrm{fid}}$ denotes the $\mathrm{SN}$ estimated by the fidelity-based witness. $3$-MUBs denote the results of the criterion by using $3$ MUBs~\cite{morelli2023resource}. $\mu_{\mathrm{fid}}$ is the $\mathrm{SN}$ estimated by the fidelity-based witness. The expectation values are calculated exactly.}
    \label{Fig3}
\end{figure}

Next, we evaluate the performance of the $\mathrm{SN}$ estimation algorithm by fixing $N=30$. Fig.~\ref{Fig3} displays results for the state $\vr_{\mathrm{iso}}^v$ with local dimensions $d=20,30,40,50$. In the worst-case scenario, as shown by the blue dots, our approach provides a more accurate estimation than the $3$-MUBs criterion for large noise strengths, while the $3$-MUBs criterion performs better for low noise strengths. The underlying mechanism is that the sampling complexity is proportional to the state's purity. Consequently, a smaller level of purity thereby allows for a more precise estimation of the fidelity. With the increase of $d$, our approach's noise parameter threshold of vanishing advantage would increase. For example, the threshold $v=0.75,0.8,0.85,0.9$ for $d=20,30,40,50$, respectively. This feature implies that the advantage of our approach becomes more evident as the local dimensions increase.
    
Fig.~\ref{Fig3} implies that a constant $N=30$ and then $2\times30=60$ projections suffice for $d=20,30,40,50$, which is consistent with the theoretical prediction. As a comparison, $3$-MUBs criterion requires $3d$ projections. Thus, our method employs fewer operations but obtains more accurate estimations for states with low purities. Moreover, as shown in Figs.~\ref{Fig3}(a), the optimal second moments criterion detects entanglement in the state $\vr_{\mathrm{iso}}^v$ for $v\geq1/\sqrt{21}\approx0.2182$~\cite{imai2021bound} while our approach achieves stronger detection for $v\geq0.1$. Additional numerical results, including the noisy two-qudit purified thermal states~\cite{li2025high}, the state $|\phi_d^+\rangle\langle\phi_d^+|$ under a random noise channel, and cases for a finite number of measurements per random operation, are provided in Appendix~\ref{AppendixD}. 

\textit{Conclusion and outlooks}---We have introduced a criterion and a classical post-processing algorithm to estimate the Schmidt number of bipartite high-dimensional states by considering the practical experimental challenges. The proposed method offers several significant advantages. First, various experimental platforms can randomly select projections that align with their specific device, as long as the random operations satisfy the properties of a unitary or orthogonal $2$-design. In contrast, the MUBs criterion requires experimental platforms to measure the state on MUBs, which is challenging for practical physical platforms such as cold-atom systems~\cite{euler2023detecting}.

Second, our results rely on the first-order measurement information of observables and require a $2$-design of unitary and orthogonal matrices. Existing high-order randomized measurements~\cite{wyderka2023probing,liu2023characterizing} involve unitary $4$-design, which construction is a challenging task~\cite{bannai2022explicit}. Finally, our criterion only requires measuring a single measurement basis, reducing the required experimental controls to a single channel independent of the local dimension $d$. Thus, our approach applies to many platforms, including integrated-optics platforms~\cite{schaeff2015experimental}, and is also practical for detecting time-bin entangled states~\cite{martin2017quantifying}.

There are several directions for future research. First, it would be interesting to generalize our work to other entanglement witnesses beyond the maximally entangled state. Next, one may study the performance of the proposed classical estimation method for high-order moments in randomized measurements~\cite{liu2022characterizing,wyderka2023complete,cieslinski2024analysing,brydges2019probing,elben2020mixed,yu2021optimal}, see generalizations in Appendix~\ref{AppendixE} and~\ref{AppendixF}. Finally, one may study the potential of our approach to the multipartite case~\cite{cobucci2024detecting}. Our results would also make it possible to experimentally characterize the entanglement of high-dimensional quantum channels~\cite{engineer2025certifying}.

The data that support the findings of this study are openly available at the following URL/DOI: https://doi.org/10.5281/zenodo.17798832.

\textit{Acknowledgment.---}We thank Armin Tavakoli, Simon Morelli, and Marcus Huber for their insightful discussions and suggestions. This work was supported by Beijing Natural Science Foundation (Grant No. Z240007), National Natural Science Foundation of China (Grants No. 12125402, No. 12534016, No. 12347152, No. 12405005, No. 12405006, No. 12075159, and No. 12171044), Quantum Science and Technology-National Science and Technology Major Project (Grant No. 2024ZD0302401 and No. 2021ZD0301500), Research Foundation for Youth Scholars of Beijing Technology and Business University, the Postdoctoral Fellowship Program of CPSF (No. GZC20230103), the China Postdoctoral Science Foundation (No. 2023M740118 and 2023M740119), the specific research fund of the Innovation Platform for Academicians of Hainan Province, the Research Foundation for Youth Scholars of Beijing Technology and Business University.

\bibliography{refs}

\appendix
\onecolumngrid
\addtocounter{observation}{-1}
\setcounter{figure}{0}
\renewcommand{\figurename}{SFigure}
\bigskip
In this Appendix, we give more details on the theoretical results of the paper. In Appendix~\ref{AppendixA}, we show the proof of Observation 1 in the main text. Appendix~\ref{AppendixB} shows the proof of the variance of the estimator for the entanglement fidelity. In Appendix~\ref{AppendixC}, we present information about the confidence interval estimation of the fidelity and the $t$-distribution. Moreover, we present an enhanced Schmidt number algorithm based on the Bootstrap resample technique~\cite{efron1992bootstrap,efron1982jackknife,efron1994introduction}. This algorithm is suitable for small samples and does not require the assumption of a population distribution. Appendix~\ref{AppendixD} exhibits additional numerical results for different high-dimensional entangled states. In Appendix~\ref{AppendixE}, we use the first-order moments to estimate a lower bound of the well-known entanglement measure. We show a randomized symmetric projection in Appendix~\ref{AppendixF} to estimate the Schmidt number, which can be seen as a randomized version of the symmetric projection introduced in Ref.~\cite{morelli2023resource}.

\section{Proof of the observation}\label{AppendixA}
Before proving the observation in the main text, we present a useful corollary showing how to calculate Haar integrals over the unitary and orthogonal groups. The proof of Corollary~\ref {Corollary} can be found in Refs. \cite{mele2024introduction,hashagen2018real,garcia2025quantum,liang2025real}.
\begin{corollary}\label{Corollary}
    Given an bipartite operator $\mathcal{A}\in\mathcal{L}(\mathds{C}^{d}\otimes\mathds{C}^{d})$, we have
    \begin{align}
        \int dU \, U^{\otimes2}\mathcal{A}U^{\dag\otimes2}&=a_{1}\eins_d\otimes\eins_d+a_{2}\mathbb{S}\label{A:eq1},\\
        \int dO \, O^{\otimes2}\mathcal{A}O^{\top\otimes2}&=c_{1}\eins_d\otimes\eins_d+c_{2}\mathbb{S}+c_{3}\mathbb{S}^{\top_B},\label{Ac1:eq2}
    \end{align}
    where the coefficients are
    \begin{align}
        \begin{cases}
            a_{1}=\frac{\tr(\mathcal{A})-d^{-1}\tr(\mathbb{S}\mathcal{A})}{d^2-1}\\
            a_{2}=\frac{\tr(\mathbb{S}\mathcal{A})-d^{-1}\tr(\mathcal{A})}{d^2-1}
        \end{cases}~\mathrm{and}~
        \begin{cases}
            c_{1}=\frac{(d+1)\tr(\mathcal{A})-\tr(\mathbb{S}\mathcal{A})-\tr(\mathbb{S}^{\top_B}\mathcal{A})}{d(d-1)(d+2)}\\
            c_{2}=\frac{-\tr(\mathcal{A})+(d+1)\tr(\mathbb{S}\mathcal{A})-\tr(\mathbb{S}^{\top_B}\mathcal{A})}{d(d-1)(d+2)}\\
            c_{3}=\frac{-\tr(\mathcal{A})-\tr(\mathbb{S}\mathcal{A})+(d+1)\tr(\mathbb{S}^{\top_B}\mathcal{A})}{d(d-1)(d+2)}
        \end{cases}
    \end{align}
\end{corollary}

Now, we prove Observation 1 in the main text.
\begin{observation}
    For any bipartite state $\vr$ of equal dimension $d$, the fidelity $F(\vr)$ is given by
    \begin{align}
        F(\vr)=(d+2)\mathcal{Q}(\vr)-(d+1)\mathcal{R}(\vr).
    \end{align}
    The Schmidt number of $\vr$ is lower bounded by
    \begin{align}
        \mu=
        \begin{cases}
            1, & F(\vr)=0,\\
            \lceil dF(\vr)\rceil, & F(\vr)\in(0,1],\\
        \end{cases}
    \end{align}
    where the $\lceil\cdot\rceil$ is ceil function.
\end{observation}
\begin{proof}
Given rank-1 observables $\mathcal{M}=|j\rangle\langle j|$, we have
\begin{align}
    \tr(\mathcal{M})=\tr(\mathcal{M}^{\otimes2 })
    =\tr(\mathbb{S}\mathcal{M}^{\otimes2 })
    =\tr(\mathbb{S}^{\top_B}\mathcal{M}^{\otimes2 })
    =\tr(\mathcal{M}^2)=1.
\end{align}
where $|j\rangle$ are computational basis of a $d$ dimension system. The moment $\mathcal{R}(\vr)$ with observable $\mathcal{M}$ is given by
\begin{align}
    \mathcal{R}(\vr)&=\int dU\tr[(U\otimes U)\vr(U^{\dag}\otimes U^{\dag})(\mathcal{M}\otimes \mathcal{M})]
    =\int dU\tr[\vr(U^{\dag}\otimes U^{\dag})(\mathcal{M}\otimes \mathcal{M})(U\otimes U)]\nonumber\\
    &=\tr\left[\vr\int dU(U\otimes U)(\mathcal{M}\otimes \mathcal{M})(U^{\dag}\otimes U^{\dag})\right]
    =\frac{\tr\left[\vr(\eins_d\otimes\eins_d+\mathbb{S})\right]}{d(d+1)}
    =\frac{1+\tr(\vr\mathbb{S})}{d(d+1)}\label{Aob1:eq2},
\end{align}
where the fourth equation holds since Eq.~(\ref{A:eq1}) in corollary~\ref{Corollary}.

The  moment $\mathcal{Q}(\vr)$ with observable $\mathcal{M}$ is given by
\begin{align}
    \mathcal{Q}&=\int dO\tr[(O\otimes O)\vr(O^{\top}\otimes O^{\top})(\mathcal{M}\otimes \mathcal{M})]
    =\tr\left[\vr\int dU(O^{\top}\otimes O^{\top})(\mathcal{M}\otimes \mathcal{M})(O\otimes O)\right]\nonumber\\
    &=\tr\left[\vr\int dU(O\otimes O)(\mathcal{M}\otimes \mathcal{M})(O^{\top}\otimes O^{\top})\right]
    =\frac{\tr\left[\vr\left(\eins_d\otimes\eins_d+\mathbb{S}+\mathbb{S}^{\top_B}\right)\right]}{d(d+2)}
    =\frac{1+\tr\left(\vr\mathbb{S}\right)+\tr\left(\vr\mathbb{S}^{\top_B}\right)}{d(d+2)}\label{Aob1:eq3},
\end{align}
where the last equation is true by using Eq.~(\ref{Ac1:eq2}) of corollary~\ref{Corollary}.
		
Based on Eqs.~(\ref{Aob1:eq2},\ref{Aob1:eq3}), we obtain
\begin{align}\label{Aob1:eq5}
    \tr(\vr\mathbb{S}^{\top_B})&=d(d+2)\mathcal{Q}(\vr)-d(d+1)\mathcal{R}(\vr).
\end{align}
		
Notice that the fidelity-based witness is $\mathcal{W}_{\mu}=\mu\eins_d^{\otimes2}/d-|\phi_d^{+}\rangle\langle\phi_d^{+}|$, where $|\phi_d^{+}\rangle=\sum_{j=0}^{d-1}|jj\rangle/\sqrt{d}$. For any state $\vr$ with Schmidt number $\leq\mu$, we have $\tr(\mathcal{W}_{\mu}\vr)\geq0$. Recall that $|\phi_d^{+}\rangle$ can be expressed as
\begin{align}
    |\phi_d^{+}\rangle\langle\phi_d^{+}|
    =\frac{1}{d}\sum_{j,k=0}^{d-1}|j\rangle\langle k|\otimes|j\rangle\langle k|
    =\frac{1}{d}\sum_{j,k=0}^{d-1}|j\rangle\langle k|\otimes(|k\rangle\langle j|)^{\top}
    =\frac{1}{d}\mathbb{S}^{\top_B}.
\end{align}
As a result, the fidelity for any state $\vr$ is
\begin{align}
    F(\vr)=\tr(\vr|\phi_d^{+}\rangle\langle\phi_d^{+}|)
    =\frac{1}{d}\tr(\vr\mathbb{S}^{\top_B})
    =(d+2)\mathcal{Q}(\vr)-(d+1)\mathcal{R}(\vr),
\end{align}
which is the first result of observation 1.

For any state $\vr$ with Schmidt number $\mu$,
\begin{align}
    \tr(\mathcal{W}_{\mu}\vr)&=\frac{\mu-\tr(\vr\mathbb{S}^{\top_B})}{d}
    =\frac{\mu}{d}-F(\vr)\geq0.
\end{align}
The above inequality also implies that for any state $\vr$ with Schmidt number $\mu$, we have
\begin{align}
    dF(\vr)\leq\mu.
\end{align}
Thus, we finish the proof.
\end{proof}

From the proof, the choice of observable depends on the moments $\mathcal{R}(\vr)$ and $\mathcal{Q}(\vr)$ which are used to estimate $\tr[\vr\mathbb{S}]$ and $\tr[\vr\mathbb{S}^{\top_B}]$, respectively. Hence, the observable must be selected satisfying $a_2\neq0$ and $c_3\neq0$ in Eqs.(\ref{A:eq1},\ref{Ac1:eq2}). Without loss of generality, we assume a diagonal observable of the form $\mathcal{M}=\mathrm{diag}(m_1,m_2,\cdots,m_r,0,\cdots,0)$ with rank $r$ and nonzero real numbers $m_1,\cdots,m_r$. These parameters can be chosen such that the conditions $a_2\neq0$ and $c_3\neq0$ hold. In practice, however, we prefer the simplest structure that satisfies the requirements, thereby easing experimental implementation. It turns out that a rank-one projection $\mathcal{M}=|j\rangle\langle j|$ suffices, where $\{|j\rangle\}_{j=0}^{d-1}$ is the computational basis in a $d$-dimensional Hilbert space. In this case, the coefficients reduce to
\begin{align}
    a_1=a_2=\frac{1}{d(d+1)},c_1=c_2=c_3=\frac{1}{d(d+2)},
\end{align}
since $\tr(\mathcal{M}^{\otimes2})=\tr(\mathbb{S}\mathcal{M}^{\otimes2})=\tr(\mathbb{S}^{\top_B}\mathcal{M}^{\otimes2})=1$.
In summary, fidelity estimation can be achieved using only a rank-one projection. This choice not only fulfills the theoretical requirements but also offers a more experimentally friendly scheme.

\section{Variance estimation}\label{AppendixB}
The fidelity is estimated via
\begin{align}
    F_{e}(\vr)=(d+2)\mathcal{Q}_e(\vr)-(d+1)\mathcal{R}_e(\vr),
\end{align}
where the estimated moments $\mathcal{R}_e(\vr)=\mathbb{E}[\mathcal{S}_U]$ and $\mathcal{Q}_e(\vr)=\mathbb{E}[\mathcal{S}_O]$ are defined as the mean values of $\mathcal{S}_U$ and $\mathcal{S}_O$, respectively. Equivalently, we obtain an estimator $F_{e}^{(l)}(\vr)=(d+2)E_O^{(l)}(\vr)-(d+1)E_U^{(l)}(\vr)$, where
\begin{align}
    E_U^{(l)}(\vr)=\tr[\vr U_l^{\otimes2}\mathcal{M}^{\otimes2 }U_l^{\dag\otimes2}],
    E_O^{(l)}(\vr)=\tr[\vr O_l^{\otimes2}\mathcal{M}^{\otimes2 }O_l^{\top\otimes2}].
\end{align}
So, the variance of $F_{e}^{(l)}(\vr)$ depends on the variance of $E_U^{(l)}(\vr)$ and $E_O^{(l)}(\vr)$, such as
\begin{align}
    \mathrm{Var}\left[F_{e}^{(l)}(\vr)\right]&=(d+2)^2\mathrm{Var}\left[E_O^{(l)}(\vr)\right]+(d+1)^2\mathrm{Var}\left[E_U^{(l)}(\vr)\right].
\end{align}
The following two subsections demonstrate that the variance of the estimator $F_{e}(\vr)$ is $\mathrm{Var}\left[F_{e}^{(l)}(\vr)\right]=\mathcal{O}(d^2)$ scaling as $\mathcal{O}(d^2)$ in the worst case. This result relies on the variance bound $\mathrm{Var}\left[E_U^{(l)}(\vr)\right]=\mathrm{Var}\left[E_O^{(l)}(\vr)\right]=\mathcal{O}(1)\tr(\vr^2)$ which is established in next two subsections. In contrast, for two specific states $\vr_{\mathrm{iso}}^{v},\vr_{\mathrm{deph}}^{v}$, we find that the variance $\mathrm{Var}\left[F_{e}^{(l)}(\vr)\right]=\mathcal{O}(1)$, and thus becomes independent of the local dimension $d$. This result is proved in the final subsection.

If $N$ experimental data are available, $F_{e}^{(1)}(\vr),F_{e}^{(2)}(\vr),\cdots,F_{e}^{(N)}(\vr)$, the unbiased estimator can be constructed from the empirical mean $\hat{F}_e(\vr)=\frac{1}{N}\sum_{l=1}^{N}F_{e}^{(l)}(\vr)$. The estimator $\hat{F}_e(\vr)$ is an unbiased estimator since
\begin{align}
    \mathbb{E}\left[\hat{F}_e(\vr)\right]=\mathbb{E}\left[\frac{1}{N}\sum_{l=1}^{N}F_{e}^{(l)}(\vr)\right]
    =\frac{1}{N}\sum_{l=1}^{N}\mathbb{E}\left[F_{e}^{(l)}(\vr)\right]
    =F(\vr),
\end{align}
where the last equation holds due to the following equation
\begin{align}
    \mathbb{E}\left[F_{e}^{(l)}(\vr)\right]&=(d+2)\mathbb{E}\left[E_O^{(l)}(\vr)\right]-(d+1)\mathbb{E}\left[E_U^{(l)}(\vr)\right]\nonumber\\
    &=(d+2)\mathbb{E}\left[\tr[\vr O_l^{\otimes2}\mathcal{M}^{\otimes2 }O_l^{\top\otimes2}]\right]-(d+1)\mathbb{E}\left[\tr[\vr U_l^{\otimes2}\mathcal{M}^{\otimes2 }U_l^{\dag\otimes2}]\right]\nonumber\\
    &=(d+2)\int dO\tr[\vr O_l^{\otimes2}\mathcal{M}^{\otimes2 }O_l^{\top\otimes2}]-(d+1)\int dU\tr[\vr U_l^{\otimes2}\mathcal{M}^{\otimes2 }U^{\dag\otimes2}]\nonumber\\
    &=(d+2)\mathcal{Q}(\vr)-(d+1)\mathcal{R}(\vr)\nonumber\\
    &=F(\vr).
\end{align}
The mean square error of the estimator $\hat{F}_e(\vr)$ depends on its variance
\begin{align}
    \mathrm{Var}\left[\hat{F}_e(\vr)\right]=\frac{1}{N^2}\sum_{l=1}^{N}\mathrm{Var}\left[F_{e}^{(l)}(\vr)\right].
\end{align}
Based on Chebyshev's inequality, the failure probability of $|\hat{F}_e(\vr)-\mathbb{E}[\hat{F}_e(\vr)]|\geq\epsilon$ for a precision $\epsilon$ has an upper bound
\begin{align}
    \mathrm{Pr}\left(|\hat{F}_e(\vr)-\mathbb{E}[\hat{F}_e(\vr)]|\geq\epsilon\right)\leq\frac{\mathrm{Var}\left[\hat{F}_e(\vr)\right]}{\epsilon^2}.
\end{align}
To make sure that the failure probability is at most $\delta$, let $\frac{\mathrm{Var}\left[\hat{F}_e(\vr)\right]}{\epsilon^2}\leq\delta$. For the general case, the variance $\mathrm{Var}\left[\hat{F}_e(\vr)\right]=\frac{\mathcal{O}(d^2)\tr(\vr^2)}{N}$ and thus the number of samples $N=\frac{\mathcal{O}(d^2)\tr(\vr^2)}{\delta\epsilon^2}$, which depends on local dimension $d$. However, for two special cases $\vr_{\mathrm{iso}}^{v},\vr_{\mathrm{deph}}^{v}$, we have $\mathrm{Var}\left[\hat{F}_e(\vr)\right]=\frac{\mathcal{O}(1)\tr(\vr^2)}{N}$ and thus the number of samples $N=\frac{\mathcal{O}(1)\tr(\vr^2)}{\delta\epsilon^2}$, which is independent of on local dimension $d$.

\subsection{Variance of unitary item}
To estimate the variance of the unitary item $E_U^{(l)}(\vr)$, we first decompose the state $\vr$ in the generalized Gell-Mann matrices basis $\{\lambda_j\}_{j=0}^{d^2-1}$~\cite{kimura2003bloch,bertlmann2008bloch} such as $\vr=\frac{1}{d^2}\sum_{j,k=0}^{d^2-1}T_{jk}\lambda_j\otimes\lambda_k$. Here, the operators $\lambda_j$ satisfying $\tr(\lambda_j)=0$ for $j=1,\cdots,d^2-1$ and $\tr(\lambda_j\lambda_k)=d\delta_{jk}$. Then, the variance of $E_U^{(l)}(\vr)$ can be expressed as
\begin{align}
    \mathrm{Var}\left[E_U^{(l)}(\vr)\right]&\leq\mathbb{E}\left[\left(E_U^{(l)}(\vr)\right)^2\right]\nonumber\\
    &=\mathbb{E}\left[\tr\left(\vr^{\otimes2} U_l^{\otimes4}\mathcal{M}^{\otimes4}U_l^{\dag\otimes4}\right)\right]\nonumber\\
    &=\tr\left[\vr^{\otimes2} \int d UU^{\otimes4}\mathcal{M}^{\otimes4}U^{\dag\otimes4}\right]\nonumber\\
    &=\frac{1}{d^4}\tr\left[\sum_{j_1,k_1=0}^{d^2-1}\sum_{j_2,k_2=0}^{d^2-1}T_{j_1k_1}T_{j_2k_2}\lambda_{j_1}\otimes\lambda_{k_1}\otimes\lambda_{j_2}\otimes\lambda_{k_2} \int d UU^{\otimes4}\mathcal{M}^{\otimes4}U^{\dag\otimes4}\right].
\end{align}
Based on the Schur-Weyl duality of the unitary group~\cite{goodman2000representations,collins2006integration,gross2007evenly,dankert2009exact,zhang2014matrix}, the operator
\begin{align}
    \int d UU^{\otimes4}\mathcal{M}^{\otimes4}U^{\dag\otimes4}
    &=\frac{24}{\Delta}C_{(4)},\Delta=d(d+1)(d+2)(d+3),
\end{align}
where the operator~\cite{zhang2014matrix}
\begin{align}
    C_{(4)}&=\frac{1}{24}P_{(1)}+\frac{1}{24}\left(P_{(12)}+P_{(13)}+P_{(14)}+P_{(23)}+P_{(24)}+P_{(34)}\right)
    +\frac{1}{24}\left(P_{(12)(34)}+P_{(13)(24)}+P_{(14)(23)}\right)\nonumber\\
    &+\frac{1}{24}\left(P_{(123)}+P_{(132)}+P_{(124)}+P_{(142)}+P_{(134)}+P_{(143)}+P_{(243)}+P_{(234)}\right)\nonumber\\
    &+\frac{1}{24}\left(P_{(1234)}+P_{(1243)}+P_{(1324)}+P_{(1342)}+P_{(1423)}+P_{(1432)}\right),
\end{align}
and $P_{\pi}$ denotes the permutation operator of the permutation $\pi$. See Example 3.36. in~\cite{zhang2014matrix}. Thus, the variance
\begin{align}
    \mathrm{Var}\left[E_U^{(l)}(\vr)\right]
    &\leq\frac{1}{d^4\Delta}\tr\left[\sum_{j_1,k_1=0}^{d^2-1}\sum_{j_2,k_2=0}^{d^2-1}T_{j_1k_1}T_{j_2k_2}\lambda_{j_1}\otimes\lambda_{k_1}\otimes\lambda_{j_2}\otimes\lambda_{k_2} \int d UU^{\otimes4}\mathcal{M}^{\otimes4}U^{\dag\otimes4}\right]\nonumber\\
    &=\frac{1}{d^4\Delta}\sum_{j_1,k_1=0}^{d^2-1}\sum_{j_2,k_2=0}^{d^2-1}T_{j_1k_1}T_{j_2k_2}\left[d^2\delta_{j_1k_1}\delta_{j_2k_2}+d^2\delta_{j_1j_2}\delta_{k_1k_2}+d^2\delta_{j_1k_2}\delta_{k_1j_2}+6\tr(\lambda_{j_1}\lambda_{k_1}\lambda_{j_2}\lambda_{k_2})\right]\nonumber\\
    &=\frac{1}{d^4\Delta}\left[d^2\left(\sum_{j_1=0}^{d^2-1}T_{j_1j_1}\right)^2+d^2\sum_{j_1,k_1=0}^{d^2-1}T_{j_1k_1}^2+d^2\sum_{j_1,k_1=0}^{d^2-1}T_{j_1k_1}T_{k_1j_1}+6\sum_{\substack{j_1,k_1\\j_2,k_2}=0}^{d^2-1}T_{j_1k_1}T_{j_2k_2}\tr(\lambda_{j_1}\lambda_{k_1}\lambda_{j_2}\lambda_{k_2})\right]\nonumber\\
    &\leq\frac{1}{d^4\Delta}\left[d^6\left[\tr(\Phi_1\vr)\right]^2+d^4\tr(\vr^2)+d^4\tr(\vr^2)+6\sum_{\substack{j_1,k_1\\j_2,k_2}=0}^{d^2-1}T_{j_1k_1}T_{j_2k_2}\tr(\lambda_{j_1}\lambda_{k_1}\lambda_{j_2}\lambda_{k_2})\right],
\end{align}
where the second inequality holds since $\sum_{j_1,k_1=0}^{d^2-1}T_{j_1k_1}^2=d^2\tr(\vr^2)$ and
\begin{align}
    \sum_{j_1,k_1=0}^{d^2-1}T_{j_1k_1}T_{k_1j_1}
    \leq\sum_{j_1,k_1=0}^{d^2-1}|T_{j_1k_1}||T_{k_1j_1}|
    \leq\sum_{j_1,k_1=0}^{d^2-1}\frac{|T_{j_1k_1}|^2+|T_{k_1j_1}|^2}{2}
    \leq\sum_{j,k=0}^{d^2-1}T_{jk}^2=d^2\tr(\vr^2).
\end{align}
Here, we introduce a pure state $\Phi_1=\frac{1}{d^2}\sum_{j=0}^{d^2-1}\lambda_j\otimes\lambda_j$ since $\tr(\Phi_1^2)=1$ and the overlap
\begin{align}
    \tr(\Phi_1\vr)=\frac{1}{d^4}\sum_{j=0}^{d^2-1}\sum_{l_1,l_2=0}^{d^2-1}T_{l_1l_2}d\delta_{jl_1}d\delta_{jl_2}
    =\frac{1}{d^2}\sum_{j=0}^{d^2-1}T_{jj},
\end{align}
which implies $\sum_{j=0}^{d^2-1}T_{jj}=d^2\tr(\Phi_1\vr)$. Now, we consider the item
\begin{align}
    \sum_{\substack{j_1,k_1\\j_2,k_2}=0}^{d^2-1}T_{j_1k_1}T_{j_2k_2}\tr(\lambda_{j_1}\lambda_{k_1}\lambda_{j_2}\lambda_{k_2})
    &\leq\sum_{\substack{j_1,k_1\\j_2,k_2}=0}^{d^2-1}\left|T_{j_1k_1}\right|\left|T_{j_2k_2}\right|\left|\tr(\lambda_{j_1}\lambda_{k_1}\lambda_{j_2}\lambda_{k_2})\right|\nonumber\\
    &\leq d^2\sum_{\substack{j_1,k_1\\j_2,k_2}=0}^{d^2-1}\left|T_{j_1k_1}\right|\left|T_{j_2k_2}\right|\nonumber\\
    &=d^2\left(\sum_{j_1,k_1=0}^{d^2-1}\left|T_{j_1k_1}\right|\right)^2\nonumber\\
    &\leq d^2\left(\sum_{j_1,k_1=0}^{d^2-1}1^2\right)\left(\sum_{j_1,k_1=0}^{d^2-1}T_{j_1k_1}^2\right)\nonumber\\
    &=d^2\times d^4\times d^2\tr(\vr^2)\nonumber\\
    &=d^8\tr(\vr^2),
\end{align}
where the first inequality holds since $|\tr(\lambda_{j_1}\lambda_{k_1}\lambda_{j_2}\lambda_{k_2})|\leq\sqrt{\tr(\lambda_{j_1}^2)\tr(\lambda_{k_1}^2)\tr(\lambda_{j_2}^2)\tr(\lambda_{k_2}^2)}=d^2$ and the second inequality uses the Cauchy-Schwarz inequality.

Thus, the variance has an upper bound
\begin{align}
    \mathrm{Var}\left[E_U^{(l)}(\vr)\right]
    &\leq\frac{1}{d^4\Delta}\left[d^6\left[\tr(\Phi_1\vr)\right]^2+d^4\tr(\vr^2)+d^4\tr(\vr^2)+6\sum_{\substack{j_1,k_1\\j_2,k_2}=0}^{d^2-1}T_{j_1k_1}T_{j_2k_2}\tr(\lambda_{j_1}\lambda_{k_1}\lambda_{j_2}\lambda_{k_2})\right]\nonumber\\
    &\leq\frac{1}{d^4\Delta}\left[d^6\left[\tr(\Phi_1\vr)\right]^2+d^4\tr(\vr^2)+d^4\tr(\vr^2)+6d^8\tr(\vr^2)\right]\nonumber\\
    &\leq\frac{d^6\left[\tr(\Phi_1\vr)\right]^2+d^4\tr(\vr^2)+d^4\tr(\vr^2)+6d^8\tr(\vr^2)}{d^5(d+1)(d+2)(d+3)}\nonumber\\
    &=\mathcal{O}(1)\tr(\vr^2).
\end{align}

By examining the scaling of each term in the numerator relative to the denominator $d^5(d+1)(d+2)(d+3)\sim d^8$, the dominant contribution for large $d$ is identified as the term proportional to $6d^8\tr(\vr^2)$. All other terms ($d^6\left[\tr(\Phi_1\vr)\right]^2$ and $d^4\tr(\vr^2)$) are suppressed by factors of $1/d^2$ and $1/d^4$ and vanish as $d\to\infty$. Therefore, the bound simplifies to $\mathcal{O}(1)\tr(\vr^2)$, which implies that the variance of the estimator $E_U^{(l)}(\vr)$ does not grow with the local dimension $d$; it is ultimately controlled only by the purity of the bipartite state.

\subsection{Variance of orthogonal item}
Similarly, the variance of $E_O^{(l)}(\vr)$ is
\begin{align}
    \mathrm{Var}\left[E_O^{(l)}(\vr)\right]&\leq\mathbb{E}\left[\left(E_O^{(l)}(\vr)\right)^2\right]\nonumber\\
    &=\mathbb{E}\left[\tr\left[\vr^{\otimes2} O_l^{\otimes4}\mathcal{M}^{\otimes4}O_l^{\dag\otimes4}\right]\right]\nonumber\\
    &=\tr\left[\vr^{\otimes2} \int d OO^{\otimes4}\mathcal{M}^{\otimes4}O^{\top\otimes4}\right]\nonumber\\
    &=\frac{1}{d^4}\tr\left[\sum_{j_1,k_1=0}^{d^2-1}\sum_{j_2,k_2=0}^{d^2-1}T_{j_1k_1}T_{j_2k_2}\lambda_{j_1}\otimes\lambda_{k_1}\otimes\lambda_{j_2}\otimes\lambda_{k_2} \int d OO^{\otimes4}\mathcal{M}^{\otimes4}O^{\dag\otimes4}\right].
\end{align}
Based on the Schur-Weyl duality for the orthogonal group~\cite{collins2006integration,west2025real,garcia2025quantum}, the operator
\begin{align}
    \int d OO^{\otimes4}\mathcal{M}^{\otimes4}O^{\top\otimes4}
    &=\sum_{x}a_{x}C_{x},
\end{align}
where the operator $C_{x}$ is a combination of the representation $F_{d}$ of the Brauer algebra $\mathcal{B}_4$ and $a_{x}$ are coefficients associated with $\tr(M),\tr(M^2),\tr(M^3),\tr(M^4)$. To derive a bound for the orthogonal item, we note that the coefficients $a_{x}=\mathcal{O}(1/d^4)$ since $\tr(M)=\tr(M^3)=\tr(M^2)=\tr(M^4)=1$~\cite{garcia2025quantum}. Then, the variance of $E_O^{(l)}(\vr)$ has an upper bound
\begin{align}
    \mathrm{Var}\left[E_O^{(l)}(\vr)\right]
    &\leq\frac{1}{d^4}\sum_{x}a_{x}\sum_{j_1,k_1=0}^{d^2-1}\sum_{j_2,k_2=0}^{d^2-1}T_{j_1k_1}T_{j_2k_2}\tr\left[\lambda_{j_1}\otimes\lambda_{k_1}\otimes\lambda_{j_2}\otimes\lambda_{k_2}C_{x}\right],
\end{align}
based on the Supplemental Theorem 6 in Ref.~\cite{garcia2025quantum}. For instance, for operators $A,B,C,D$, the trace is
\begin{eqnarray}
\tr\left[\left(A\otimes B\otimes C\otimes D\right)C_1\right]=
\begin{tabular}{c}
    \includegraphics[scale=0.15]{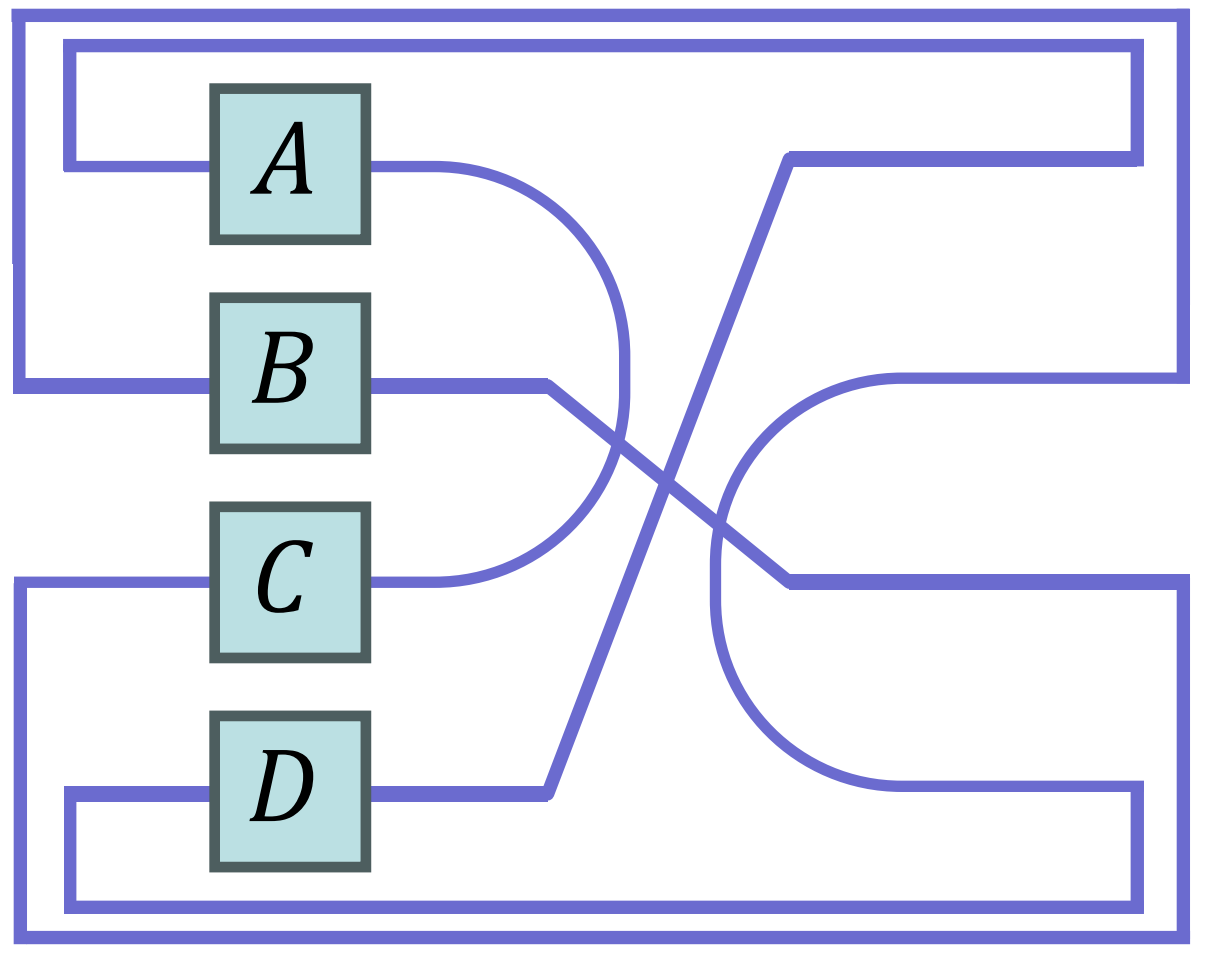}
\end{tabular}
=\tr\left(DBCA^{\top}\right),
\end{eqnarray}
where, the operator $C_1$ can be seen as $\{1,3\}\{2,7\}\{4,5\}\{6,8\}$. Moreover, the number of parameters $a_{x}$ is $\frac{8!}{2^44!}=105$ independent of local dimension $d$. Notice that the quantity $\left|\tr\left[\lambda_{j_1}\otimes\lambda_{k_1}\otimes\lambda_{j_2}\otimes\lambda_{k_2}C_{x}\right]\right|\leq\left|\tr\left(\lambda_{j_1}\lambda_{k_1}\lambda_{j_2}\lambda_{k_2}\right)\right|\leq d^2$. This can be demonstrated using tensor-network diagrams for matrices and operations $C_{x}$. See examples in Supplemental Figure 1 of the work~\cite{garcia2025quantum}. The variance has an upper bound
\begin{align}
    \mathrm{Var}\left[E_O^{(l)}(\vr)\right]
    &\leq\frac{1}{d^4}\sum_{x}a_{x}\sum_{j_1,k_1=0}^{d^2-1}\sum_{j_2,k_2=0}^{d^2-1}T_{j_1k_1}T_{j_2k_2}\tr\left[\lambda_{j_1}\otimes\lambda_{k_1}\otimes\lambda_{j_2}\otimes\lambda_{k_2}C_{x}\right]\nonumber\\
    &\leq\frac{1}{d^4}\sum_{x}a_{x}\sum_{j_1,k_1=0}^{d^2-1}\sum_{j_2,k_2=0}^{d^2-1}|T_{j_1k_1}||T_{j_2k_2}|d^2\nonumber\\
    &\leq\frac{1}{d^4}\sum_{x}a_{x}d^8\tr(\vr^2)\nonumber\\
    &\leq105\tr(\vr^2)=\mathcal{O}(1)\tr(\vr^2),
\end{align}
where the last inequality holds since $a_{x}=\mathcal{O}(1/d^4)$~\cite{garcia2025quantum}.

\subsection{Special cases for different states}
The above bound is the worst case for general states. Here, we consider two important cases:
\begin{align}
    \vr_{\mathrm{iso}}^{v}&=v|\phi_d^+\rangle\langle\phi_d^+|+[(1-v)/d^2]\eins_d^{\otimes2},\nonumber\\
    \vr_{\mathrm{deph}}^{v}&=v|\phi_d^+\rangle\langle\phi_d^+|+[(1-v)/d]\sum_{j=0}^{d-1}|jj\rangle\langle jj|,
\end{align}
where the parameter $v\in[0,1]$.

In particular, the correlation matrices of the states $\vr_{\mathrm{iso}}^{v}$ and $\vr_{\mathrm{deph}}^{v}$ are diagonal matrices. Thus, the following inequality is true,
\begin{align}
    \sum_{\substack{j_1,k_1\\j_2,k_2}=0}^{d^2-1}\left|T_{j_1k_1}\right|\left|T_{j_2k_2}\right|\left|\tr(\lambda_{j_1}\lambda_{k_1}\lambda_{j_2}\lambda_{k_2})\right|
    &=\sum_{j_1,j_2=0}^{d^2-1}\left|T_{j_1j_1}\right|\left|T_{j_2j_2}\right|\left|\tr(\lambda_{j_1}\lambda_{j_1}\lambda_{j_2}\lambda_{j_2})\right|\nonumber\\
    &\leq d^2\sum_{j_1,j_2=0}^{d^2-1}\left|T_{j_1j_1}\right|\left|T_{j_2j_2}\right|\nonumber\\
    &=d^2\left(\sum_{j_1=0}^{d^2-1}\left|T_{j_1j_1}\right|\right)^2\nonumber\\
    &\leq d^2\left(\sum_{j_1=0}^{d^2-1}1^2\right)\left(\sum_{j_1=0}^{d^2-1}T_{j_1j_1}^2\right)\nonumber\\
    &=d^2\times d^2\times d^2\tr(\vr^2)\nonumber\\
    &=d^6\tr\left(\vr^2\right).
\end{align}
The above inequality implies that the variance $\mathrm{Var}\left[E_O^{(l)}(\vr)\right]=\mathrm{Var}\left[E_U^{(l)}(\vr)\right]=\mathcal{O}(1/d^2)\tr(\vr^2)$. Then, the variance of the estimator $F_{e}^{(l)}(\vr)$ reduces to $\mathrm{Var}\left[F_{e}^{(l)}(\vr)\right]=\mathcal{O}(1)$, which is independent of local dimension $d$.

More generally, we consider the random noise state
\begin{align}
    \vr_{\mathrm{rand}}^{v}(\sigma)&=v|\phi_d^+\rangle\langle\phi_d^+|+(1-v)\sigma,
\end{align}
where $v\in[0,1]$ and the random states $\sigma$ is sampled according to the Hilbert-Schmidt metric~\cite{zyczkowski2005average}. From the derivation, we note that the dominate item in the variance estimation for $\vr_{\mathrm{rand}}^{v}(\sigma)$ is the quantity $\mathcal{L}=\sum_{j_1,k_1=0}^{d^2-1}\left|\hat{T}_{j_1k_1}\right|$, where $\hat{T}_{j_1k_1}$ is the correlation matrix of $\sigma$. This claim is true since
\begin{align}
    \sum_{j_1,k_1=0}^{d^2-1}\left|vT_{j_1k_1}+(1-v)\hat{T}_{j_1k_1}\right|
    &\leq\sum_{j_1,k_1=0}^{d^2-1}v\left|T_{j_1k_1}\right|+(1-v)\left|\hat{T}_{j_1k_1}\right|\nonumber\\
    &=v\sum_{j_1,k_1=0}^{d^2-1}\left|T_{j_1k_1}\right|+(1-v)\sum_{j_1,k_1=0}^{d^2-1}\left|\hat{T}_{j_1k_1}\right|\nonumber\\
    &=vd^2+(1-v)\sum_{j_1,k_1=0}^{d^2-1}\left|\hat{T}_{j_1k_1}\right|,
\end{align}
where $T_{j_1k_1}=\{1,-1\}$ is the correlation matrix of the state $|\phi_d^+\rangle\langle\phi_d^+|$.

To see the scaling of $\mathcal{L}$ with local dimension $d$, we numerically simulate $10^3$ random states $\sigma,$ and we find that $\mathcal{L}\leq d^2$. This bound implies the variance of the state $\vr_{\mathrm{rand}}^{v}(\sigma)$ is also $\mathcal{O}(1)$ independent of local dimension $d$.

\subsection{Variance estimation with finite samples}
In the previous section, we assumed the number of samples for estimating the expectation value of observables is infinite. However, this assumption is impossible because practical cases always obtain a finite number of samples to estimate the expectation values. In a practical situation, a finite number of samples is available.

Now, to estimate the fidelity $F(\vr)$, the total number of measurements is denoted as $N_{\mathrm{tot}}=N\times K$. Here, $N$ is the number of random unitaries or orthogonal matrices, and $K$ is the number of measurements for a fixed unitary or orthogonal matrix. For the $l$-th unitary/orthogonal, the expectation value $E_U^{(l)}(\vr)=\tr\left[\vr U^{\otimes2}\mathcal{M}^{\otimes2}U^{\dag\otimes2}\right]$ is estimated on a quantum device. To do this, we measure the expectation value $\mathcal{M}^{\otimes2}$ on $K$ copies of $\vr$ and record all $K$ measurement outcome $X_1^{(l)},\cdots,X_{K}^{(l)}$. We construct an unbiased estimation of the expectation value $E_U^{(l)}(\vr)$ as $\hat{X}^{(l)}=\frac{1}{K}\sum_{k=1}^{K}X_k^{(l)}$ since
\begin{align}
    \mathbb{E}_{X}\left[\hat{X}^{(l)}\right]
    =\mathbb{E}_{X}\left[\frac{1}{K}\sum_{k=1}^{K}X_k^{(l)}\right]
    =\frac{1}{K}\sum_{k=1}^{K}\mathbb{E}_{X}\left[X_k^{(l)}\right]
    =\frac{1}{K}\sum_{k=1}^{K}E_U^{(l)}(\vr)
    =E_U^{(l)}(\vr).
\end{align}
The variance of the estimator $\hat{X}^{(l)}$ is
\begin{align}
    \mathrm{Var}\left[\hat{X}^{(l)}\right]
    &=\frac{1}{K^2}\sum_{k=1}^{K}\mathrm{Var}\left[X_{k}^{(l)}\right]
    =\frac{KE_U^{(l)}(\vr)\left[1-E_U^{(l)}(\vr)\right]}{K^2}
    =\frac{E_U^{(l)}(\vr)\left[1-E_U^{(l)}(\vr)\right]}{K}.
\end{align}

The unbiased estimator for the moment $\mathcal{R}(\vr)$ can be given by
\begin{align}
    \hat{\mathcal{R}}=\frac{1}{N}\sum_{i=1}^{N}\hat{X}^{(l)}
    =\frac{1}{NK}\sum_{i=1}^{N}\sum_{k=1}^{K}X_k^{(l)}.
\end{align}
The unbiasedness of this estimator can be verified since
\begin{align}
    \mathbb{E}_{U}\mathbb{E}_{X}\left[\hat{\mathcal{R}}\right]
    &=\mathbb{E}_{U}\mathbb{E}_{X}\left[\frac{1}{NK}\sum_{i=1}^{N}\sum_{k=1}^{K}X_k^{(l)}\right]
    =\mathcal{R}(\vr).
\end{align}

Now, we derive the variance of the estimator $\hat{\mathcal{R}}$. By the definition of variance, we have
\begin{align}
    \mathrm{Var}\left[\hat{\mathcal{R}}\right]
    &=\mathbb{E}_{U}\mathbb{E}_{X}\left[\hat{\mathcal{R}}^2\right]-\mathbb{E}_{U}\mathbb{E}_{X}\left[\hat{\mathcal{R}}\right]^2\nonumber\\
    &=\mathbb{E}_{U}\mathbb{E}_{X}\left[\frac{1}{N^2K^2}\sum_{l_1,l_2=1}^{N}\sum_{k_1,k_2=1}^{K}X_{k_1}^{(l_1)}X_{k_2}^{(l_2)}\right]-\left[\mathcal{R}(\vr)\right]^2\nonumber\\
    &=\frac{1}{N^2K^2}\sum_{l_1,l_2=1}^{N}\sum_{k_1,k_2=1}^{K}\mathbb{E}_{U}\mathbb{E}_{X}\left[X_{k_1}^{(l_1)}X_{k_2}^{(l_2)}\right]-\left[\mathcal{R}(\vr)\right]^2.
\end{align}

For $l_1=l_2=l$, the first term is
\begin{align}
    \mathrm{first~term}
    &=\frac{1}{N^2K^2}\sum_{l=1}^{N}\sum_{k_1,k_2=1}^{K}\mathbb{E}_{U}\mathbb{E}_{X}\left[X_{k_1}^{(l)}X_{k_2}^{(l)}\right]\nonumber\\
    &=\frac{1}{N^2}\sum_{l=1}^{N}\mathbb{E}_{U}\left[\mathbb{E}_{X}\left[\left(\frac{1}{K}\sum_{k_1=1}^{K}X_{k_1}^{(l_1)}\right)^2\right]\right]\nonumber\\
    &=\frac{1}{N^2}\sum_{l=1}^{N}\mathbb{E}_{U}\left[\frac{E_U^{(l)}(\vr)}{K}+\frac{K-1}{K}\left[E_U^{(l)}(\vr)\right]^2\right]\nonumber\\
    &=\frac{1}{NK}\mathcal{R}(\vr)+\frac{K-1}{NK}\int dU\left[E_U(\vr)\right]^2
\end{align}
The item
\begin{align}
    \int dU\left[E_U(\vr)\right]^2
    &=\int \tr\left[U^{\otimes2}\vr U^{\dag\otimes2}\mathcal{M}^{\otimes 2}\otimes U^{\otimes2}\vr U^{\dag\otimes2}\mathcal{M}^{\otimes 2}\right]\nonumber\\
    &=\tr\left[\vr^{\otimes2}U^{\otimes 2}\mathcal{M}^{\otimes2}U^{\dag\otimes2}\right]
    =\mathcal{O}(1)\tr(\vr^2),
\end{align}
where the last equation has been proved in the last subsection.

For $l_1\neq l_2$, the first term is
\begin{align}
    \mathrm{first~term}
    &=\frac{1}{N^2K^2}\sum_{l_1\neq l_2=1}^{N}\sum_{k_1,k_2=1}^{K}\mathbb{E}_{U}\mathbb{E}_{X}\left[X_{k_1}^{(l_1)}X_{k_2}^{(l_2)}\right]\nonumber\\
    &=\frac{1}{N^2}\sum_{l_1\neq l_2=1}^{N}\mathbb{E}_{U}\left[\mathbb{E}_{X}\left[\sum_{k_1=1}^{K}\frac{1}{K}X_{k_1}^{(l_1)}\right]\mathbb{E}_{X}\left[\sum_{k_2=1}^{K}\frac{1}{K}X_{k_2}^{(l_2)}\right]\right]\nonumber\\
    &=\frac{1}{N^2}\sum_{l_1\neq l_2=1}^{N}\mathbb{E}_{U}\left[E_{U}^{(l_1)}(\vr)E_{U}^{(l_2)}(\vr)\right]\nonumber\\
    &=\frac{N-1}{N}\mathbb{E}_{U}\left[E_{U}^{(l_1)}(\vr)E_{U}^{(l_2)}(\vr)\right]\nonumber\\
    &=\frac{N-1}{N}\mathbb{E}_{U}\left[E_{U}^{(l_1)}(\vr)\right]\mathbb{E}_{U}\left[E_{U}^{(l_2)}(\vr)\right]\nonumber\\
    &=\frac{N-1}{N}\left[\mathcal{R}(\vr)\right]^2.
\end{align}

Hence, the variance of $\hat{\mathcal{R}}$ is
\begin{align}
    \mathrm{Var}\left[\hat{\mathcal{R}}\right]
    &=\frac{1}{NK}\mathcal{R}(\vr)+\frac{K-1}{NK}\mathcal{O}(1)\tr(\vr^2)+\frac{N-1}{N}\left[\mathcal{R}(\vr)\right]^2-\left[\mathcal{R}(\vr)\right]^2\nonumber\\
    &\leq\frac{1}{NK}\mathcal{R}(\vr)+\frac{K-1}{NK}\mathcal{O}(1)\tr(\vr^2).
\end{align}

Similarly, we can calculate the variance of the orthogonal part associated with the moment $\mathcal{Q}(\vr)$. Based on the analysis of the last subsection, we find that the variance $\mathrm{Var}\left[\hat{\mathcal{Q}}\right]\leq\frac{1}{NK}\mathcal{Q}(\vr)+\frac{K-1}{NK}\mathcal{O}(1)\tr(\vr^2)$, where $\hat{\mathcal{Q}}$ is an unbiased estimation of the moment $\mathcal{Q}(\vr)$. In summary, the unbiased estimation of the fidelity can be constructed by $\hat{F}=(d+2)\hat{\mathcal{Q}}-(d+1)\hat{\mathcal{R}}$. We find the variance of the estimator $\hat{F}$ is
\begin{align}
    \mathrm{Var}\left[\hat{F}\right]
    &=(d+2)^2\mathrm{Var}\left[\hat{\mathcal{Q}}\right]+(d+1)^2\mathrm{Var}\left[\hat{\mathcal{R}}\right]\nonumber\\
    &<\frac{(d+2)^2}{NK}\mathcal{Q}(\vr)+\frac{(d+2)^2(K-1)}{NK}\mathcal{O}(1)\tr(\vr^2)+\frac{(d+1)^2}{NK}\mathcal{R}(\vr)+\frac{(d+1)^2(K-1)}{NK}\mathcal{O}(1)\tr(\vr^2)\nonumber\\
    &=\frac{(d+2)^2}{NK}\mathcal{Q}(\vr)+\frac{(d+1)^2}{NK}\mathcal{R}(\vr)+2\frac{(d+2)^2(K-1)}{NK}\mathcal{O}(1)\tr(\vr^2)\nonumber\\
    &=\frac{2d+3+(2d+3)\tr(\vr\mathbb{S})+(d+1)\tr(\vr\mathbb{S}^{\top_B})}{dNK}+2\frac{(d+2)^2(K-1)}{NK}\mathcal{O}(1)\tr(\vr^2).
\end{align}
Here, we introduce a pure state $\Phi_1=\frac{1}{d^2}\sum_{j=0}^{d^2-1}\lambda_j\otimes\lambda_j$ since $\tr(\Phi_1^2)=1$ and the overlap $\tr(\Phi_1\vr)=\frac{1}{d^2}\sum_{j=0}^{d^2-1}T_{jj}
    =d^{-1}\tr(\vr\mathbb{S})$, which implies $\tr(\vr\mathbb{S})=d\tr(\Phi_1\vr)$. The trace $\tr(\vr\mathbb{S}^{\top_B})=d\tr(\vr|\phi_d^{+}\rangle\langle\phi_d^{+}|)$. Thus, the variance turns to
\begin{align}
    \mathrm{Var}\left[\hat{F}\right]
    &=\frac{2d+3+(2d+3)\tr(\vr\mathbb{S})+(d+1)\tr(\vr\mathbb{S}^{\top_B})}{dNK}+2\frac{(d+2)^2(K-1)}{NK}\mathcal{O}(1)\tr(\vr^2)\nonumber\\
    &=\frac{2d+3+(2d+3)d\tr(\Phi_1\vr)+(d+1)d\tr(\vr|\phi_d^{+}\rangle\langle\phi_d^{+}|)}{dNK}+2\frac{(d+2)^2(K-1)}{NK}\mathcal{O}(1)\tr(\vr^2)\nonumber\\
    &=\frac{2+3/d+(2d+3)\tr(\Phi_1\vr)+(d+1)\tr(\vr|\phi_d^{+}\rangle\langle\phi_d^{+}|)}{NK}+2\frac{(d+2)^2(K-1)}{NK}\mathcal{O}(1)\tr(\vr^2)\nonumber\\
    &<\frac{6+3/d+3d}{NK}+2\frac{(d+2)^2(K-1)}{NK}\mathcal{O}(1)\tr(\vr^2)\nonumber\\
    &=\frac{\mathcal{O}(3d)}{NK}+\frac{\mathcal{O}(d^2)(K-1)}{NK}\tr(\vr^2).
\end{align}
In the limit of an infinite number of measurements per random operation, we have $\lim_{K\to\infty}\left[\frac{\mathcal{O}(3d)}{NK}+\frac{\mathcal{O}(d^2)(K-1)}{NK}\tr(\vr^2)\right]=\frac{\mathcal{O}(d^2)}{N}\tr(\vr^2)$ which matches the result in the previous subsections. The total number of measurements required is $N_{\mathrm{tot}}=NK=\mathcal{O}(d^2/\delta\epsilon^2)$. Since $K$ is associated with the precision of expectation values of the observable, we care about the number of random operations. As a result, $N=\mathcal{O}(d^2)$ for general states, which has the same scaling as the direct fidelity estimation.

Now, consider the special case where the state is a maximally entangled state under different noise channels. In this case, the variance
\begin{align}
    \mathrm{Var}\left[\hat{F}\right]
    &=\frac{\mathcal{O}(3d)}{NK}+\frac{K-1}{NK}\tr(\vr^2).
\end{align}
Again, in the limit of an infinite number of measurements per random operation, we have $\lim_{K\to\infty}\left[\frac{\mathcal{O}(3d)}{NK}+\frac{K-1}{NK}\tr(\vr^2)\right]=\frac{1}{N}\tr(\vr^2)$, which also matches the result in the previous subsection. Since $K$ is associated with the precision of expectation values of the observable, we care about the number of random operations. As a result, $N=\mathcal{O}(1)$ for these special states by letting $K=\mathcal{O}(d)$. Compared with direct fidelity estimation $\mathcal{O}(d^2)$ and the MUBs method $\mathcal{O}(md)$, the randomized projection method in this work achieves $\mathcal{O}(1)$ for maximally entangled states under different noises, making it more efficient for high-dimensional quantum states.

Note that in the proof of the variance we assume sampled random unitaries are unitary $4$-design and random orthogonal matrices are orthogonal $4$-design.

\section{Confidence interval for the mean and the bootstrapping
resample technique}\label{AppendixC}
Suppose we are interested in estimating a population parameter, such as the mean $g$ of a population, from a random sample. Let $X = \{X_j\}_{j=1}^N$ be a sample of $N$ independent and identically distributed observations drawn from the population. The population distribution is unknown, and we do not impose any prior assumption about its form.

A common point estimator for the population mean $g$ is the sample mean $\bar{X}=\mathbb{E}[x]=\frac{1}{N}\sum_{j=1}^{N}X_j$, which is an unbiased estimation of the mean $g$. When the population variance is unknown, the confidence interval for $g$ is constructed by
\begin{align}\label{AppendixC:eq1}
    \bar{X}\pm t_{\alpha,N-1}\frac{s}{\sqrt{N}}.
\end{align}
Here, $s=\sqrt{\frac{\sum_{j=1}^{N}\left(X_j-\bar{X}\right)^2}{N-1}}$ is the sample standard deviation, and the critical factor $t_{\alpha,N-1}$ comes from the $t$-distribution. Note that $t_{\alpha,N-1}$ depends on the confidence level $\alpha$ and sample size $N$. For example, $t_{95\%,11}\neq t_{95\%,20}$.

The central limit Theorem~\cite{rosenblatt1956central,bellhouse2001central} states that for larger $N$ (often empirically taken as $N\geq30$), the sample mean $\mathbb{E}[x]$ is approximately normally distributed, with mean $g$ and standard deviation $\sigma(\bar{X})=\frac{\sigma}{\sqrt{N}}$. Regardless of the distribution of the population, as the sample size is increased, the sampling distribution of the sample mean becomes a normal distribution. That is, we can always obtain a confidence interval with sample size $N\geq30$, and the population distribution is not important.

For a small sample size $N<30$, the confidence interval is only efficient when the population distribution is normal. If the assumption is violated, the constructed confidence interval may be inaccurate. However, the $t$ distribution-based confidence interval is relatively robust to this assumption~\cite{montgomery2010applied}.

Two of the traditional methods for obtaining a confidence interval for the mean of a non-normal distribution are the central limit theorem~\cite{rosenblatt1956central,bellhouse2001central}, reported above, and the bootstrap resampling technique~\cite{efron1992bootstrap,efron1982jackknife,efron1994introduction}. The bootstrap resampling technique is a popular non-parametric method and is useful for small sample sizes. The detailed process is the following~\cite{efron1992bootstrap,efron1982jackknife,efron1994introduction,wang2001confidence,pek2017confidence}.

\textbf{Input:} Given a data set $X=\{X_j\}_{j=1}^N$ with $N$ samples randomly generated from the population. Generally, $N<30$.

\textit{Step 1.} Resample the observed sample with replacement and calculate the
sample mean for this bootstrap sample.

\textit{Step 2.} Repeat \textit{Step 1.} $B$ times.

\textbf{Output:} Construct the confidence interval for the data set with size $B$ by using the Eq.~(\ref{AppendixC:eq1}).

Using the above knowledge, we can estimate the fidelity in the main text. In particular, we express the fidelity as
\begin{align}
    F(\vr)&=(d+2)\mathcal{Q}(\vr)-(d+1)\mathcal{R}(\vr)
    =(d+2)\int dOE_O(\vr)-(d+1)\int dUE_U(\vr)
    =\int dWF_W(\vr),
\end{align}
where the quantity $F_W(\vr)=(d+2)E_O(\vr)-(d+1)E_U(\vr)$. We can treat the fidelity $F(\vr)$ as a mean of an unknown population. Thus, we can estimate the fidelity $F(\vr)$ from a statistical view. After sampling $N$ random unitaries and random orthogonal matrices, we obtain $N$ fidelities estimated by a single unitary and a single orthogonal matrix. Except for obtaining a point estimation from finite samples, we also try to construct a confidence interval of the unknown mean, $F(\vr)$. The construction approach follows Eq.~(\ref{AppendixC:eq1}).

Note that the population distribution underlying the mean $F(\vr)$ is unknown. Hence, it is unclear whether it follows a normal distribution. Empirically, the direct use of Eq.~(\ref{AppendixC:eq1}) for constructing a confidence interval typically requires a sample size $N\geq30$. In our setting, however, obtaining a large number of samples is challenging because each measurement involves implementing a local high-dimensional unitary or orthogonal operation, which is experimentally demanding. We therefore aim to extract the maximum information from a limited set of experimental observations. Consequently, we prefer to work with a small number of random operations, typically $N<30$. When the population distribution is unknown, confidence intervals built from such small samples may, in principle, be inaccurate. Nevertheless, as noted earlier, the $t$-distribution-based interval is relatively robust to departures from normality~\cite{montgomery2010applied}. As supported by SFig.~\ref{SFig1}, this direct construction performs satisfactorily even for $N<30$.

\begin{figure}[ht]
    \centering
    \includegraphics[scale=0.5]{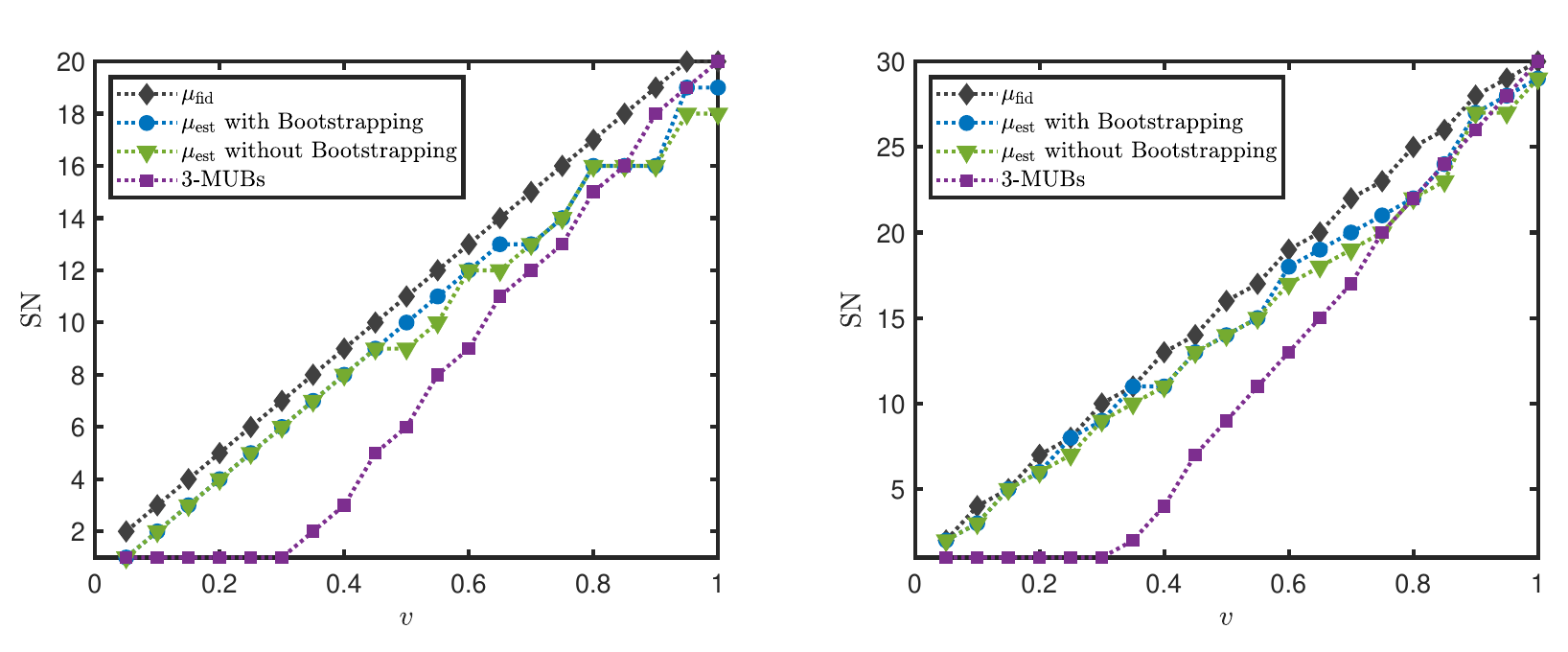}
    \caption{Test for algorithm based on Bootstrapping resampling technique. Numerical results for the states $\vr_{\mathrm{iso}}^v$ with dimension (a) $d=20$ and (b) $d=30$ by fixing $N=15$ and the CL $99.9\%$. We set the number of resampling $B=2000$.}
    \label{SFig1}
\end{figure}

\section{Analysis of numerical results}\label{AppendixD}
This section reviews the $3$-MUBs criterion~\cite{morelli2023resource} and the randomized measurements~\cite{liu2023characterizing,wyderka2023probing} for detecting the Schmidt number.

The $3$-MUBs criterion~\cite{morelli2023resource} employes $3$-MUBs to estimate the $\mathrm{SNs}$. For any dimensions, reference~\cite{li2025high} constructs the following $3$ MUBs, $\{|e_a^1\rangle=|a\rangle\}_{a=0}^{d-1}$, $\{|e_a^2\rangle\}_{a=0}^{d-1}$, and $\{|e_a^3\rangle\}_{a=0}^{d-1}$ with
\begin{align}
    |e_a^2\rangle&=\frac{1}{\sqrt{d}}\sum_{j=0}^{d-1}e^{2\pi i\left(\frac{\alpha j^2}{2d}+\frac{aj}{2d}\right)}|j\rangle,
    ~
    |e_a^3\rangle=\frac{1}{\sqrt{d}}\sum_{j=0}^{d-1}e^{2\pi i\left(\frac{(\alpha+d-1)j^2}{2d}+\frac{aj}{2d}\right)}|j\rangle.
\end{align}
Based on the criterion~\cite{morelli2023resource}, for any bipartite state $\vr$ with dimension $d$ and $\mathrm{SN}$ at most $\mu$ it hold that $\mathcal{S}_{3,d}(\vr)\leq1+\frac{2\mu}{d}$, where the quantity
\begin{align}\label{SII:eq1}
    \mathcal{S}_{3,d}(\vr)=\sum_{z=1}^{3}\sum_{a=0}^{d-1}\tr\left[\vr|e_a^z\rangle\langle e_a^z|\otimes|e_a^{z*}\rangle\langle e_a^{z*}|\right].
\end{align}
    
It is clear from Eq.~(\ref{SII:eq1}) that the total number of projections required to perform $3$ MUBs is $3d$, which scales linearly with dimension $d$. In contrast, our approach requires $2N$ projections, where $N$ represents the number of random unitary or orthogonal matrices.

The randomized measurements obtain the quantity $R^{(t)}=\int dU_A\int dU_BE_{AB}^t$, where $t$ is an integer, and the expectation value is defined as
\begin{align}
    E_{AB}=\tr\left[(U_A\otimes U_B)\vr(U_A^{\dag}\otimes U_B^{\dag})(G_A\otimes G_B)\right],
\end{align}
for some well-defined observables $G_A$ and $G_B$. The random unitary $U_A$ and $U_B$ are sampled via the Haar measure on the unitary group~\cite{goodman2000representations,collins2006integration,gross2007evenly,dankert2009exact}.

Second, we consider the decomposition of $\vr$ in terms of the Gell-Mann matrices
\begin{align}
    \vr=\frac{1}{d^2}&\left[\eins_d\otimes\eins_d+\sum_{j=1}^{d^2-1}\alpha_j\lambda_j\otimes\eins_d+\sum_{k=1}^{d^2-1}\eins_d\otimes\beta_k\lambda_k+\sum_{j,k=1}^{d^2-1}T_{jk}\lambda_j\otimes\lambda_k\right].
\end{align}
In this way, the reduced states are
\begin{align}
    \vr_A=\tr_B(\vr)=\frac{1}{d}\left(\eins_d+\sum_{j=1}^{d^2-1}\alpha_j\lambda_j\right),
    \vr_B=\tr_B(\vr)=\frac{1}{d}\left(\eins_d+\sum_{k=1}^{d^2-1}\beta_k\lambda_k\right).
\end{align}
We define the correlation matrix $T=(T_{jk})$ for $j,k=1,\cdots,d^2-1$. In particular, we have
\begin{align}
    \tr(\vr^2)&=\frac{1}{d^2}\Big(1+\sum_{j=1}^{d^2-1}\alpha_j^2+\sum_{k=1}^{d^2-1}\beta_k^2+\sum_{j,k=1}^{d^2-1}T_{jk}^2\Big),\nonumber\\
    \tr(\vr_{A}^2)&=\frac{1}{d}\Big(1+\sum_{j=1}^{d^2-1}\alpha_j^2\Big),
    \tr(\vr_{B}^2)=\frac{1}{d}\Big(1+\sum_{k=1}^{d^2-1}\beta_k^2\Big),
\end{align}
and then the purity is
\begin{align}
    \tr(\vr^2)=\frac{1}{d^2}\Big[d\tr(\vr_{A}^2)+d\tr(\vr_{B}^2)-1+\sum_{j,k=1}^{d^2-1}T_{jk}^2\Big].
\end{align}
It is clear to see that
\begin{align}
    \|T\|_{\tr}^2&=\sum_{j,k=1}^{d^2-1}T_{jk}^2
    =d^2\tr(\vr^2)-d\tr(\vr_{A}^2)-d\tr(\vr_{B}^2)+1.
\end{align}
Note that each purities $\tr(\vr^2)$, $\tr(\vr_{A}^2)$, and $\tr(\vr_{B}^2)$ can be accessed by the second moments $R^{(2)}$.

The criterion based on the second moments is constructed by calculating the second moments of the $d\otimes d$ state $|\phi_{+}^{\mu}\rangle=\frac{1}{\sqrt{\mu}}\sum_{j=0}^{\mu-1}|jj\rangle$. For $\mu=1,2,\cdots,d$, we obtain a data set
\begin{align}
    \left[R^{(2)}(|\phi_{+}^{1}\rangle\langle\phi_{+}^{1}|),\cdots,R^{(2)}(|\phi_{+}^{d}\rangle\langle\phi_{+}^{d}|)\right].
\end{align}
Then given a state $\vr$ if $R^{(2)}(\vr)\leq R^{(2)}(|\phi_{+}^{\mu}\rangle\langle\phi_{+}^{\mu}|)$ for $\mu\in[1,d-1]$ then $\vr$ has a Schmidt number $\mu$. If $R^{(2)}(\vr)> R^{(2)}(|\phi_{+}^{d-1}\rangle\langle\phi_{+}^{d-1}|)$, then $\vr$ has a Schmidt number $d$. See details in Refs.~\cite{imai2021bound,liu2023characterizing,wyderka2023probing}. We call this criterion $2$-RMs as shown in SFig.~\ref{SFig2}(a).

The purified thermal states mixed with white noise are
\begin{align}
    \vr_{\mathrm{thermal}}(v,\beta)=v|\psi(\beta)\rangle\langle\psi(\beta)|+\frac{v}{d^2}\eins_{d^2},
\end{align}
where $v\in[0,1]$ and the pure state
\begin{align}
    |\psi(\beta)\rangle=\frac{1}{\sqrt{\sum_{j=0}^{d-1}e^{-\beta j}}}\sum_{j=0}^{d-1}e^{-\beta j}|jj\rangle.
\end{align}
Furthermore, we consider the random noise state
\begin{align}
    \vr_{\mathrm{rand}}^{v}(\sigma)&=v|\phi_d^+\rangle\langle\phi_d^+|+(1-v)\sigma,
\end{align}
where $v\in[0,1]$ and the random states $\sigma$ is sampled according to the Hilbert-Schmidt metric~\cite{zyczkowski2005average}.

Finally, we simulate the performance of the $\mathrm{SN}$ estimation algorithm in an experimental scenario. Note that in the main text, we iteratively run the $\mathrm{SN}$ estimation algorithm $N_{\mathrm{iter}}$ times for a fixed $N$. However, practically iterating many runs is challenging. Thus, we are capable of obtaining $2N$ expectation values by sampling $N$ random unitaries and $N$ random orthogonal matrices. Then, we estimate the $\mathrm{SN}$ by using the obtained experimental data and the estimation algorithm in the main text. SFig.~\ref{SFig2}, SFig.~\ref{SFig3}, and Table~\ref{table1} display additional numerical results for different states. These results suggest that $N=30$ is efficient for high dimensions and exhibits significant advantages compared with the $3$-MUBs criterion and the second-order moment from randomized measurements.

\begin{figure}[t]
    \centering
    \includegraphics[scale=0.5]{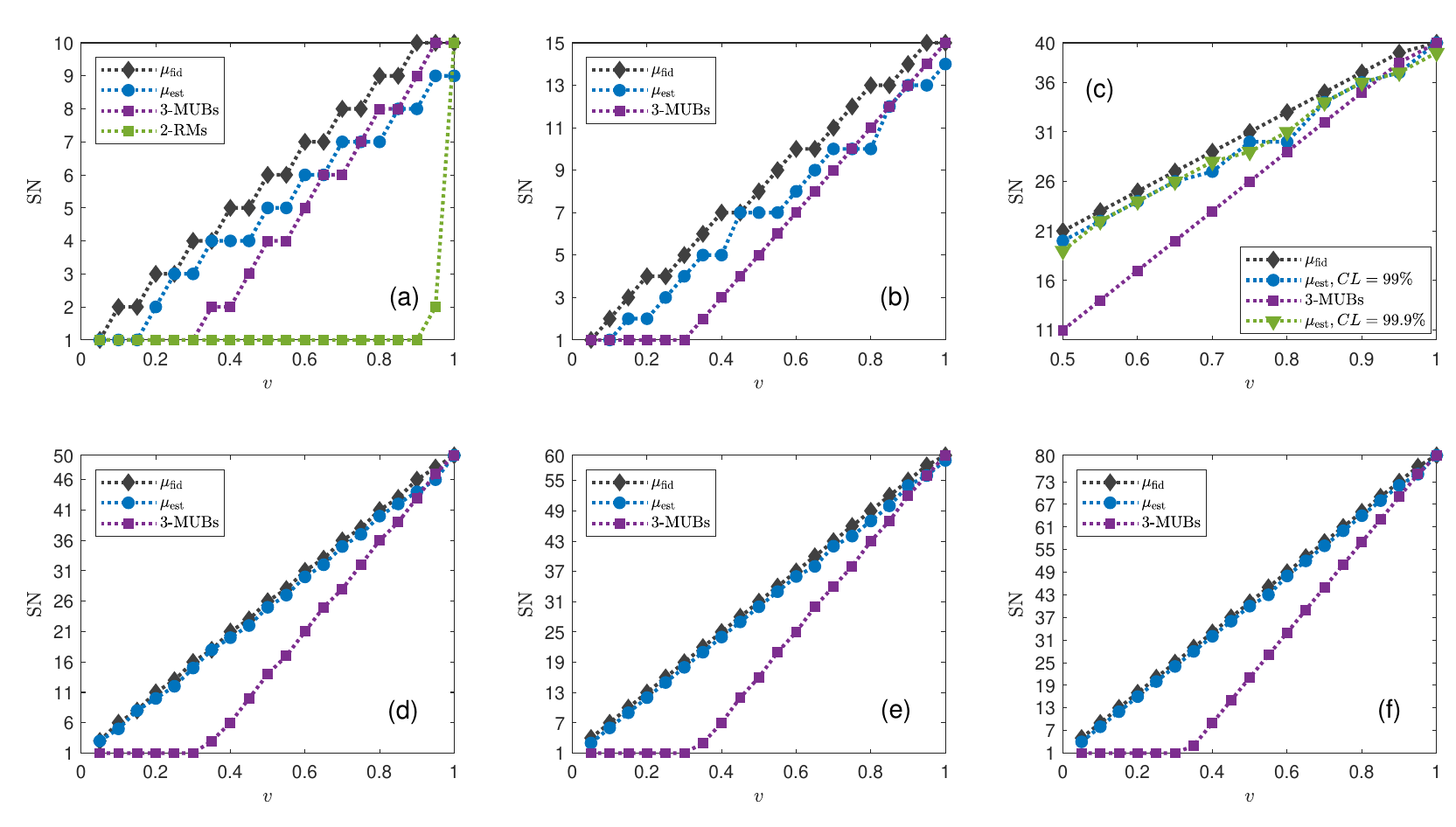}
    \caption{Numerical results for the states $\vr_{\mathrm{iso}}^v$ with dimension (a) $d=10$, (b) $d=15$, (c) $d=40$, (d) $d=50$, (e) $d=60$, (f) $d=80$. For (a,b), the number of measurements for each product observable is $K=1000$. For (c-f), we directly calculate the expectation values of observables. In (a,b,d,e,f), the CL is $99\%$.}
    \label{SFig2}
\end{figure}

\begin{table}[htbp]
\centering
\caption{Comparison between $3$-MUBs criterion and our approach for detecting Schmidt number of the state $\vr_{\mathrm{rand}}^{v}(\sigma)$. Here, $N=30$, the parameter $v\in[0,1]$ and the state $\sigma$ are random for each iteration. The number of measurements for each product observable is $K=1000$, and the CL is $99\%$. We randomly sample $1000$ random states $\sigma$ and exhibit the percentages.}
\label{table1}
\begin{tabular}{|c|c|c|c|}
\hline
 & Our method is better & Equivalent performance & $3$-MUBs is better \\ 
\hline
$d=20$ & $76\%$ & $14.8\%$ & $9.2\%$ \\
\hline
$d=30$ & $85.3\%$ & $9.3\%$ & $5.4\%$ \\
\hline
$d=40$ & $88\%$ & $8.2\%$ & $3.8\%$ \\
\hline
\end{tabular}
\end{table}

\begin{figure}[t]
    \centering
    \includegraphics[scale=0.5]{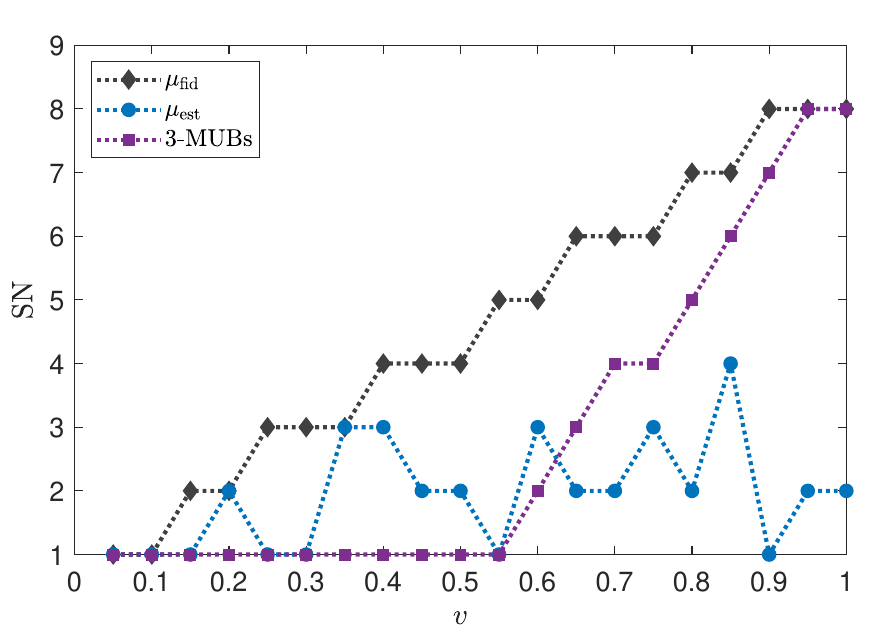}
    \caption{The results for the noisy two-qudit purified thermal states $\vr_{\mathrm{thermal}}(v,\beta)$ with dimension $d=20$ and the parameter $\beta=0.5$. We exactly calculate the expectation values for each product observable, and the CL is $99.9\%$.}
    \label{SFig3}
\end{figure}

The code for generating random unitaries and orthogonal matrices can be found in the link: "http://www.qetlab.com/RandomUnitary".

\section{The entanglement quantification}\label{AppendixE}
Here, we show that the moments $\mathcal{R}(\vr)$ and $\mathcal{Q}(\vr)$ induce lower bounds of entanglement measures in quantum information.
	
An important entanglement measure is based on the trace distance, $E_{T}(\vr)=\min_{\sigma\in\mathrm{SEP}}D_{T}(\vr,\sigma)$, given by the smallest trace distance between $\vr$ and separable state set $\mathrm{SEP}$~\cite{nielsen2010quantum}, where the trace distance is defined as half of the trace norm, $D_{T}(\vr,\sigma)=\frac{1}{2}\|\vr-\sigma\|_{\tr}$, and $\|C\|_{\tr}$ denotes the trace norm of $C$, i.e., the sum of the absolute values of eigenvalues of $C$. We present our results as follows.
\begin{observation}\label{ob2}
    For any bipartite state $\vr$ of equal dimension $d$, the entanglement measure based on the trace distance $E_T(\vr)$ has a lower bound,
    \begin{align}
		E_T(\vr)\geq\frac{1}{d}-(d+2)\mathcal{Q}(\vr)+(d+1)\mathcal{R}(\vr).
	\end{align}
\end{observation}
\begin{proof}
    Let $\sigma_{\mathrm{opt}}$ be the closest separable state. We have the inequality, $E_T(\vr)=\frac{1}{2}\|\vr-\sigma_{\mathrm{opt}}\|_1\geq\frac{1}{2}|\tr[(\vr-\sigma_{\mathrm{opt}})]|\geq\frac{1}{2}|\tr[(\vr-\sigma_{\mathrm{opt}})\mathcal{W}_1^{'}]|\geq-\tr(\vr\mathcal{W}_1)$, where $\mathcal{W}_1=\eins_d^{\otimes2}/d-|\phi_d^{+}\rangle\langle\phi_d^{+}|$ is the fidelity-based entanglement witness~\cite{guhne2021geometry} and $\mathcal{W}_1^{'}=2\mathcal{W}_1-(\frac{2}{d}-1)\eins_d^{\otimes^2}$. The first inequality holds due to the absolute value inequality. The second inequality is true since the absolute values of the eigenvalues of $\mathcal{W}_1^{'}$ are equal to one. The third inequality is based on the property of the entanglement witness $\mathcal{W}_1$. By expressing $\tr(\vr\mathcal{W}_1)$ in terms of $\mathcal{R}(\vr)$ and $\mathcal{Q}(\vr)$ based on observation 1,
    \begin{align}
        \tr(\vr\mathcal{W}_1)&=\frac{1}{d}-\tr(\vr|\phi_d^{+}\rangle\langle\phi_d^{+}|)
        =\frac{1}{d}-F(\vr)
        =\frac{1}{d}-(d+2)\mathcal{Q}(\vr)+(d+1)\mathcal{R}(\vr),
    \end{align}
    we complete the proof.
\end{proof}

Following the Refs.~\cite{zhang2016evaluation,sun2024bounding}, other entanglement measures such as the concurrence~\cite{wootters1998entanglement}, the G-concurrence~\cite{gour2005family}, the entanglement of formation~\cite{wootters2001entanglement}, the geometric measure of entanglement~\cite{wei2003geometric}, and the robustness of entanglement~\cite{vidal1999robustness} have a connection with the entanglement witness $\mathcal{W}_1$. Thus, these entanglement measures can also be bounded by a function of $\mathcal{R}(\vr)$ and $\mathcal{Q}(\vr)$.

\section{Randomized symmetric projections}\label{AppendixF}
In the main text, we propose a method to estimate the fidelity with the maximally entangled state using random unitaries and orthogonal matrices. Here, we note that the estimation of fidelity can be achieved by using only random unitaries.

Given an observable $\mathcal{M}=|j\rangle\langle j|$ and random unitary $U$, the 
integral
\begin{align}
    \int dUU^{\otimes2}\mathcal{M}^{\otimes2}U^{\dag\otimes2}
    &=a_{1}\eins_d\otimes\eins_d+a_{2}\mathbb{S}
    =\frac{\eins_d\otimes\eins_d+\mathbb{S}}{d(d+1)},
\end{align}
where the last equation uses the Eq.(\ref{Ac1:eq2}). Thus, $\mathrm{SWAP}$ operator can be rewritten as
\begin{align}
    \mathbb{S}&=d(d+1)\int dUU^{\otimes2}\mathcal{M}^{\otimes2}U^{\dag\otimes2}-\eins_d\otimes\eins_d
    =d(d+1)\int dU|e_j\rangle\langle e_j|\otimes|e_j\rangle\langle e_j|
    -\eins_d\otimes\eins_d,
\end{align}
where the pure state $|e_j\rangle=U|j\rangle$. For a bipartite state $\vr$, the fidelity is
\begin{align}
    F(\vr) & = \tr(\vr|\phi_d^{+}\rangle\langle\phi_d^{+}|)
    =\frac{1}{d}\tr(\vr\mathbb{S}^{\top_B})\nonumber\\
    &=(d+1)\int dU\tr\left[\vr U\mathcal{M}U^{\dag}\otimes\left(U\mathcal{M}U^{\dag}\right)^{\top}\right]- \frac{1}{d}\nonumber\\
    &=(d+1)\int dU\tr\left[\vr U\mathcal{M}U^{\dag}\otimes U^{*}\mathcal{M}^{\top}U^{T}\right]- \frac{1}{d}\nonumber\\
    &=(d+1)\int dU\tr\left[\vr U\mathcal{M}U^{\dag}\otimes U^{*}\mathcal{M}U^{T}\right]- \frac{1}{d}\nonumber\\
    &= (d+1)\int dU\tr\left[\vr\left(|e_j\rangle\langle e_j|\otimes|e_j\rangle\langle e_j|\right)^{\top_B}\right] - \frac{1}{d}\nonumber\\
    &= (d+1)\int dU\tr\left[\vr\left(|e_j\rangle\langle e_j|\otimes|e_j^{*}\rangle\langle e_j^{*}|\right)\right] - \frac{1}{d},\label{E:eq1}
\end{align}
where $|e_j^{*}\rangle$ denotes the complex conjugate of $|e_j\rangle$. Hence, we require to perform $U\otimes U^{*}$ on $\vr$ and then measure the product observable $|j\rangle\langle j|^{\otimes2}$. This process equivalents to perform random symmetric projections $|e_j\rangle\langle e_j|\otimes|e_j^{*}\rangle\langle e_j^{*}|$ on $\vr$. The measured results can be used to estimate the fidelity by Eq.(\ref{E:eq1}) and then infer the Schmidt number. In practice, the fidelity can be estimated by a finite number of random samples
\begin{align}
    F(\vr) = (d+1)\sum_{l=1}^{N}\tr\left[\vr\left(|e_j^{l}\rangle\langle e_j^{l}|\otimes|e_j^{l*}\rangle\langle e_j^{l*}|\right)\right]-\frac{1}{d}.
\end{align}
Remark that this method has also been proposed in Ref.~\cite{liu2022characterizing}. Compared with our strategy in the main text, we use the same random operations, which can be achieved by using a fixed measurement channel. Moreover, we use random orthogonal matrices, which are only real-value matrices and thus are more easily performed experimentally. Interestingly, Eq.(\ref{E:eq1}) introduces a random symmetric projection to detect the Schmidt number. In this view, the MUB criterion~\cite{morelli2023resource} is a special case of symmetric projections. Thus, it would be interesting to investigate the performance of randomized symmetric projections.

\end{document}